\AddToHook{package/before/amsmath}{

}

\documentclass{article}

\usepackage{arxiv}

\usepackage{amsthm}
\usepackage{amsfonts}
\usepackage{amssymb}
\usepackage{amsmath}

\usepackage[utf8]{inputenc} 
\usepackage[T1]{fontenc}    
\usepackage{esvect}
\usepackage{epstopdf}		
\usepackage{hyperref}       
\usepackage{url}            
\usepackage{booktabs}       
\usepackage{nicefrac}       
\usepackage{graphicx}
\graphicspath{{./Figures/}}

\usepackage{doi}

\usepackage{comment}
\usepackage{dsfont}
\usepackage[symbol]{footmisc}

\usepackage{natbib}
\bibpunct[, ]{[}{]}{;}{a}{}{,}
\renewcommand\bibfont{\fontsize{10}{12}\selectfont}

\usepackage{algorithm}
\usepackage{algorithmicx}
\usepackage[noend]{algpseudocode}
\usepackage[labelsep=space]{caption}

\theoremstyle{plain}
\newtheorem{theorem}{Theorem}[section]
\newtheorem{lemma}[theorem]{Lemma}

\theoremstyle{definition}
\newtheorem{definition}[theorem]{Definition}

\theoremstyle{remark}
\newtheorem{remark}{Remark}

\newcommand*\Let[2]{\State #1 $\gets$ #2}
\algrenewcommand\alglinenumber[1]{
    {\sf\footnotesize\addfontfeatures{Colour=888888,Numbers=Monospaced}#1}}
\algrenewcommand\algorithmicrequire{\textbf{input}}

\title{Residual U-net with Self-Attention Network for Multi-Agent Time-Consistent Optimal Trade Execution}

\author{{Andrew S. Na}\thanks{CONTACT Andrew S. Na. Email: andrew.na@uwaterloo.ca}\hspace{1.5mm}\\
	David R. Cheriton School of Computer Science\\
	University of Waterloo\\
	Waterloo, ON\\
	\texttt{andrew.na@uwaterloo.ca} \\
	\And
	{Justin W.L. Wan} \\
	David R. Cheriton School of Computer Science\\
	University of Waterloo\\
	Waterloo, ON\\
	\texttt{justin.wan@uwaterloo.ca}
}

\begin{document}

\maketitle

\begin{abstract}
In this paper, we explore the use of a deep residual U-net with self-attention to solve the the continuous time time-consistent mean variance optimal trade execution problem for multiple agents and assets. Given a finite horizon we formulate the time-consistent mean-variance optimal trade execution problem following the Almgren-Chriss model as a Hamilton-Jacobi-Bellman (HJB) equation. The HJB formulation is known to have a viscosity solution to the unknown value function. We reformulate the HJB to a backward stochastic differential equation (BSDE) to extend the problem to multiple agents and assets. We utilize a residual U-net with self-attention to numerically approximate the value function for multiple agents and assets which can be used to determine the time-consistent optimal control. In this paper, we show that the proposed neural network approach overcomes the limitations of finite difference methods. We validate our results and study parameter sensitivity. With our framework we study how an agent with significant price impact interacts with an agent without any price impact and the optimal strategies used by both types of agents. We also study the performance of multiple sellers and buyers and how they compare to a holding strategy under different economic conditions.
\end{abstract}


\section{Introduction}

The optimal trade execution model was developed as a framework that would allow traders to control the speed of their portfolio liquidation in order to limit the disadvantages of liquidating an asset too quickly or too slowly \cite{almgren-2000}. This is done to limit the fluctuation of prices due to large volume movements. However, the framework for optimal trade execution under the mean variance criteria is time inconsistent due to the variance term which is non-monotonic with respect to wealth \cite{basak-2010}. If the optimal control given at time $t$ is time-inconsistent then it does not mean that the control will remain optimal at a later time $t+\Delta t < T$. Thus, the optimal portfolio execution problem is usually solved for time $t$ and the portfolio is not changed until the portfolio matures at $T$. This strategy is referred to as the \textit{pre-commitment} strategy.  

The pre-commitment solution to the mean-variance problem often reformulates the original mean-variance problem to an equivalent formulation \citep{li-2000,forsyth-2011,bjork-2017}. Pre-commitment strategy works well within stable economic regimes and is shown to outperform time-consistent strategies \citep{forsyth-2020,vigna-2020}. However, the pre-commitment assumes that there is complete information from $t$ to $T$ but in reality the state of the economy might change drastically within a given time period due to exogenous factors. The absence of time consistency also prevents us from using dynamic programming directly.

To overcome the time inconsistency problem of the mean variance criteria, there has been increasing development in continuous time consistent mean variance asset allocation in both reinsurance and investments \citep{guan-2022}. The current approaches typically solve the continuous time optimal portfolio allocation problem \citep{zhou-2000,aivaliotis-2018,aivialotis-2014,wang-2011}. The approach in \citep{van-staden-2018} considers the control of impulse controls in jump processes. 

In discrete time, this problem has been solved using dynamic programming assuming local optimality \citep{basak-2010}. The discrete time solution does not solve for the value function of the Hamilton-Jacobi-Bellman (HJB) equation. The analytical solution to the continuous time consistent mean variance asset allocation problem has been studied for an agent with $d$-assets \citep{zhou-2000} and in the $1$ asset reinsurance and investment setting using $K$ agents \citep{guan-2022}. Note, the optimal asset allocation problem can be viewed as a special case of the optimal trade execution problem when there is no price impact, i.e.  $\kappa_t = 0$ and $\kappa_P = 0$ for agents who can only sell/buy the asset. The numerical solutions to the optimal trade execution problem has been presented for the time-inconsistent strategy using finite difference methods (FDM) \citep{forsyth-2011,wang-2011}. The FDM solution to the time-consistent mean variance asset allocation problem, i.e. control on $\alpha$, and its extension to impulse control has also been solved \citep{van-staden-2018}.

Recent advances in reinforcement learning (RL) uses model free methods to solve the optimal trade execution on limit-order-books and real data \citep{ning-2021,lin-2021,chen-2022}. A RL framework that solves the continuous time mean-variance portfolio optimization problem derives the optimal policy iteration \citep{wang-2020}. However, this only considers the control of the portfolio allocation and not the optimal trade rate. A RL method has been used to approximate the optimal trade execution under mean variance criteria in the discrete time setting \citep{hendricks-2014}. Typically, RL methods use the model free (Q-learning) approach. However, the issue with using Q-learning in the optimal trade execution in discrete time is that the problem is solved locally and does not guarantee the solution given satisfies the HJB equation \citep{basak-2010}. In current frameworks, there is no clear environment where agents can interact with. The agents learn independently based on a general simulated market. In this paper, we propose a multi-agent model which couples each agent $k$ implicitly through the price of the asset. This creates an environment that takes into account the decisions made by each individual agent and allows us to simulate the market dynamics more precisely.

In addition, none of the existing numerical solutions has formally addressed the problem of solving the time-consistent mean variance optimal trade execution problem in continuous time for higher dimensions. The analytical solutions are restricted to the control on the portfolio allocation, $\alpha$, and cannot be extended to the control of the trade rate, $v$, because the trade rate is a non-linear control.  This motivates the need for numerical solutions to the continuous time time-consistent mean variance optimal trade execution problems in higher dimensions.

In this paper, we extend the optimal trade execution formulation to the time-consistent framework and the use of neural networks to solve the high-dimensional optimal trade execution problem for $d$ assets and $K$ agents. The contributions of this paper are summarized as follows:
\begin{enumerate} 
    \item We present a time-consistent optimal trade execution framework for multiple assets and multiple agents. This framework allows us to couple multiple agents through the stochastic stock price process indirectly using price impact, and to create an environment where multiple agents can interact with. 
    \item We formulate the mean variance optimal trade execution problem such that the investor is aware of other investors in the market. We reformulate the mean variance problem into its dual representation and an auxillary HJB equation. To extend this framework into higher dimensions, we reformulate the HJB equation as a backward stochastic differential equation (BSDE). The auxillary HJB equation has been shown to have a viscosity solution \citep{aivaliotis-2018}.
    \item We present a neural network solution for the $d$-dimensional assets and $K$-agent optimal trade execution problem using neural networks, which extends the numerical solution of the time consistent optimal trade execution problem to high dimensions. We use techniques based on neural network pricing of high dimensional options to formulate a deep residual self-attention network that overcomes the need to save weights at each timestep \citep{e-2017,chen-2019,han-2016,hure-2020}.
    \item We show numerical experiments on the interactions of $K$-agent optimal trade execution under different market conditions in a finite horizon setting.
\end{enumerate}


In Section 2, we describe the optimal trade execution problem and extend it to the $d$-dimensional setting. The extension of the optimal trade execution to $K$ agents with agents that are aware of the relative performance of other agents. In Section 3, we formulate at the HJB PDE representation of the $d$-dimensional, $K$ agent optimal trade execution problem, its equivalent BSDE formulation and its approximate discretization. In Section 4, we discretize the BSDE to the least squares problem and introduce the neural network models used to solve them. In Section 5, we present numerical results and conclude in Section 6.

\section{Mean-Variance Optimal Trade Execution}\label{sec:OptTrdExec}

In this section, we present the optimal trade execution problem with $i=1,...,d$ assets and $K$-agents. This is an extension of the classical framework \citep{forsyth-2011}. Suppose each agent $k$, where $k\in[1,...,K]$, can observe a basket of $d$ stocks in the market; with price process $\underline{S}_k(t) = [S_{1,k}(t), ..., S_{d,k}(t)]^\top \in \mathbb{R}^d$. Then the number of shares of the $d$ stocks owned by agent $k$ is given by $\underline{\alpha}_k(t) = [\alpha_{1,k}(t),...,\alpha_{d,k}(t)]$ and the amount invested in a risk free bank account is given by $b_k \in \mathbb{R}$. Let $t \in [0,T]$ be the time up to maturity $T$, then at any time $t$, the market agent $k$ has a wealth $X_k(t)$ given by:
\begin{equation}
    X_k(t) = b_k(t) + \sum\limits_{i=1}^{d}\alpha_{i,k}(t) S_{i,k}(t).
\label{eq:portfolio}
\end{equation}
To handle buying and selling cases symmetrically, we start off with $b_0 \in \mathbb{R}$ in cash, $\alpha_0 > 0$ shares if we are selling and $\alpha_0 < 0$ shares if buying. This means that we are trying to liquidate a long (short) position if we are selling (buying). More precisely
\begin{align}
    b_k(0)&=b_0, \;\; S_{i,k}(0)=s_0, \;\; \alpha_{i,k}(0) = \alpha_0, \label{eq:bc_0}\\
    b_k(T)&=b_T, \;\; \alpha_{i,k}(T) = 0,\label{eq:bc_T}
\end{align}
where $b_T - b_0$ is the cash generated by selling/buying in $[0,T)$. The objective is to maximize the wealth at maturity, $T$, and simultaneously minimize the risk as measured by the variance. Note that due to the non-monotonicity of the variance term, this problem as presented is time inconsistent \citep{basak-2010,tse-2012}.

A natural optimization criteria to consider for optimal trade execution is the mean-variance (MV) strategy of Markowitz. The MV strategy looks for the optimal strategy that maximizes an investors expected gain, and at the same time minimizes their risk \citep{wang-2011}. 

In this section we define the mean-variance criteria for optimal trade execution. Let $\mathbb{E}[X_k(T)|v_k(t)]$ be the conditional expected wealth given the trade rate $v_k(t)$. We denote $\mathbb{E}_{v}[X_k(T)] = \mathbb{E}[X_k(T)|v_k(t)]$. The agent $k$ with trade rate $v_k(t)$ maximizes the expected wealth and minimizes the variance defined as
\begin{equation*}
    Var_{v}[X_k(T)] = \mathbb{E}_{v}[X_k(T)^2] - (\mathbb{E}_{v}[X_k(T)])^2.
\end{equation*}
This can be reduced to a single cost function given by
\begin{equation}
    J(X_k(T)) = \mathbb{E}_{v}[X_k(T)] - \frac{\gamma_k}{2}Var_{v}(X_k(T)),
\label{eq:cost_mv}
\end{equation}
where $\gamma_k > 0$ is the risk preference of the investor $k$. Then the mean-variance optimal execution problem is to determine the strategy $v(t)$ such that:
\begin{equation}
    \textnormal{arg}\max\limits_{v} J(X_k(T)) = \textnormal{arg}\max\limits_{v}\left\{\mathbb{E}_{v}[X_k(T)] - \frac{\gamma_k}{2}Var_{v}(X_k(T))\right\}.
\label{eq:cost_mv_opt}
\end{equation}
Let $v_{k}^{*}(t)$ be the optimal control that maximizes \eqref{eq:cost_mv_opt}, then the optimal cost function is given by
\[
    J^*(X_k(T)) = \mathbb{E}_{v^*}[X_k(T)] - \frac{\gamma_k}{2}Var_{v^*}(X_k(T)).
\]

In general, the solution for the mean-variance problem is time inconsistent. Many approaches have been explored to overcome this limitation by solving an auxiliary problem which is referred to as the time consistent mean-variance problem \citep{zhou-2000,forsyth-2011,  tse-2012,aivialotis-2014,van-staden-2018,forsyth-2020,guan-2022}.

\subsection{Problem Formulation}

The dynamics of the optimal trade execution problem can be formulated by the Almgren-Chriss model \citep{almgren-2000}. The instantaneous trade rate is given by:
\[
    v(t) = \frac{d\alpha(t)}{dt},
\]
where $\alpha(t)$ is the amount of stock held by a liquidating agent. Let $\mu$ be the drift, $r$ be the risk free rate and $\sigma$ be the volatility of the asset. The price process of the risky asset follows the Geometric Brownian motion (GBM) \citep{forsyth-2011}:
\[
    dS(t) = (\mu - r + g(v(t))) S(t) dt + \sigma S(t) dW(t),
\]
where $g(v(t))$ is the level of permanent price impact. We use the following linear form for the permanent price impact \citep{gatheral-2010, forsyth-2011}:
\begin{equation*}
    g(v(t)) = \kappa_p v(t),
\end{equation*}
where the constant $\kappa_p$ is a constant permanent price impact factor.

The stock price we receive from trading is given by the price we actually receive from a trade is given by \citep{forsyth-2011}:
\[
    \bar{S}(t,v(t)) = S(t) f(v(t)),
\]
where $f(v(t))$ is the level of temporary price impact. The temporary price impact is used in our model to capture liquidity effects on the trade \citep{almgren-2012}. The function $f(v(t))$ is assumed to have the form
\begin{equation*}
    f(v(t)) = [1+\kappa_{s} \textnormal{sign}(v(t))]\textnormal{exp}\{\kappa_\tau v(t)^\beta\},
\end{equation*}
where $\kappa_{s}$ is the bid-ask spread for the asset, $\kappa_\tau$ is a constant that represents the temporary price impact factor of the investor and $\beta$ is the price impact exponent \citep{forsyth-2011}.

The cash position is then given by \citep{forsyth-2011}:
\[
    \frac{db(t)}{dt} = rb(t) - v(t) S(t) f(v(t)).
\]
Given states $\{S(T^-),B(T^-),\alpha(T^-)\}$ at the instant $t = T^-$ before the end of the trading horizon, we have one final liquidation (if $\alpha(T^-)\neq 0$) to ensure that the number of shares $\alpha(T) = 0$. The liquidation value is computed as follows \citep{tse-2012}:
\begin{align*}
    B(T) &= B(T^-) + \lim_{v\rightarrow-\infty}\{\alpha(T^-)\bar{S}(T^-,v(T^-))\}\\
    &= \int_{t=0}^{T^-}-e^{r(T-t)}v(t)\bar{S}(t,v(t))dt + \lim_{v\rightarrow-\infty}\{\alpha(T^-)\bar{S}(T^-,v(T^-)\}.
\end{align*}

\subsection{Extension of the Almgren-Chriss model to $d$ Assets and $K$ agents}

In this section, we describe how we extend the dynamics of the optimal trade execution problem to $d$ assets and $K$ agents. First, we consider $d$ assets and one agent. Let the trade rate for the agent with $d$ assets be given by $\underline{v}(t) = [v_1(t),...,v_d(t)]^\top$. As with the one dimensional case, we assume the process $S_{i}(t)$ follows the GBM, with a modification due to the permanent price impact. Let $\mu_i$ and $\sigma_i$ be the drift and volatility of stocks $i=1,...,d$. Let $\mathbf{\rho} \in \mathbb{R}^{d\times d}$ be the correlation matrix and we define the correlated random variable $dW_{i}(t) = \sum\limits_{j=1}^{d}L_{i,j}Z_{j} \sqrt{dt}$, where $Z_{j}\sim N(0,1)$ are i.i.d, and $\mathbf{L}$ is the Cholesky factor of $\mathbf{\rho}$, i.e. $\mathbf{\rho}=\mathbf{LL}^\top$. 

Given a vector of initial prices, $s_{i}(0)\in\mathbb{R}$, for each asset $i$, we have:

\begin{equation}
    dS_{i}(t) = (\mu_i - r + g_i(v_i(t))) S_{i}(t) dt + \sigma_i S_{i}(t) dW_{i}(t),
    \label{eq:SDE}
\end{equation}
with initial condition $S_{i}(0) = s_{i}(0)$.

For simplicity, we assume that the market has sufficient liquidity and we do not consider market frictions. A multi-asset optimal trade execution problem has been modelled by a Orstein-Uhlenbeck (OU) process \citep{bergault-2022}. In our model, we assume a constant $\mu_i$, $r$ and $\sigma_i$ and we introduce a permanent price impact function $g_i(v_i(t))$ in (\ref{eq:SDE}). We use the following linear form for the permanent price impact \citep{forsyth-2011}:
\begin{equation*}
    g_i(v_i(t)) = \kappa_{p,i} v_i(t),
\end{equation*}
where the constant $\kappa_{p,i}$ is the permanent price impact factor for asset $i$. A linear permanent price impact function is used to eliminate round-trip arbitrage opportunities as shown in \citep{tse-2012}. Given a vector of stock prices $\underline{s}=[s_{1},...,s_{d}]^\top$, the bank account $b$ is assumed to follow
\begin{equation}
    \frac{db(t)}{dt} = rb(t) - \sum\limits_{i=1}^{d}v_i(t)S_{i}(t)f_i(v_i(t)),
    \label{eq:bank}
\end{equation}
where $f_i(v_i(t))$ is the temporary price impact of the investor as they sell/buy the asset $i$. The function $f_i(v_i(t))$ is assumed to have the form
\begin{equation*}
    f_i(v_i(t)) = [1+\kappa_{s,i} \textnormal{sign}(v_i(t))]\textnormal{exp}\{\kappa_{\tau,i}v_i(t)^\beta\},
\end{equation*}
where $\kappa_{s,i}$ is the bid-ask spread for each asset, $\kappa_{\tau,i}$ is the temporary price impact factor of the investor for each asset and $\beta$ is the price impact exponent \citep{forsyth-2011}.

Next, we extend the optimal trade execution problem to $K$-agents. Let $k = 1,...,K$, then the trading rate for each agent $k$ is given by $\underline{v}_k = [v_{1,k},...,v_{d,k}]^\top$. Let $S_i(t)$ be the price of asset $i$ that follows the GBM. For multiple agents the dynamics of asset $i$ is given by:
\[
    dS_{i}(t) = (\mu_i - r + \sum_{k=1}^{K}g_{i,k}(v_{i,k}(t))) S_{i}(t) dt + \sigma_i S_{i}(t) dW_{i}(t),
\]
where $g_{i,k}(v_{i,k}) = \kappa_{p,i,k}v_{i,k}$. The permanent price impact over multiple agents impact the price of asset $i$ as they liquidate their position in the asset. This couples the actions of each agent $k$ implicitly through the price dynamics. Note, since our economy or state space is modeled through asset prices. This also means the economy has full information about each agent $k$. This tells us that the asset price is driven by not just the drift of our asset but also by the trade activity of each agent $k$. This is consistent with the observation that in a multiple agent setting the total market impact on an asset $i$ is the sum of permanent price impact of all agents and the temporary price impact \citep{Schoneborn-2008}.

The cash dynamic for each agent $k$ is given by
\[
    \frac{db_{k}(t)}{dt} = rb_{k}(t) - \sum\limits_{i=1}^{d}v_{i,k}(t)S_{i}(t) f_{i,k}(v_{i,k}(t)),
\]
where
\[
    f_{i,k}(v_{i,k}(t)) = [1+\kappa_{s,i} \textnormal{sign}(v_{i,k}(t))]\textnormal{exp}\{\kappa_{\tau,i,k}v_{i,k}(t)^{\beta_k}\}.
\]
Note, the bid-ask spread $\kappa_{s,i}$ is only dependent on the asset $i$. The temporary price impact factors of each agent k only affects the execution price of agent $k$, thus there is no coupling involved in this term with other agents. This gives us a measure of how each agent perceives their own impact on the market.

\subsection{Mean Variance Criteria of Performance Aware Agents}

We can extend our framework to consider agents that are aware about their performance relative to others. We refer to these agents as performance conscious agents of degree $\phi$. More formally, for agent $k$, let the market average wealth be given by \citep{guan-2022}:
\[
    \bar{X}(t) = \frac{1}{K}\sum\limits_{k=1}^{K}X_k(t).
\]
Then for $\phi \in [0,1)$ we let $\hat{X}(t) = X_k(t) - \phi \bar{X}(t)$ and the mean-variance criteria becomes \citep{guan-2022}:
\begin{equation}
    J(\hat{X}(T);v,\phi) = \mathbb{E}_{v}[\hat{X}(T)] - \frac{\gamma_k}{2}Var_{v}[\hat{X}(T)],
\label{eq:cost_mv_agent}
\end{equation}
where $\gamma_k > 0$ is the risk aversion parameter of each agent $k$.  Note that by setting $\phi=0$ in \eqref{eq:cost_mv_agent}, the cost function reduces to \eqref{eq:cost_mv}. The solution to the $K$-agent problem results in an optimal control \citep{guan-2022}.
The Nash equilibrium has been shown in \citep{guan-2022} and is given by the following definition:

\begin{definition}\label{def:optimal_control}
Optimal control of time consistent mean-variance problem for $K$-agents \citep{guan-2022}. Let $v_k\in\mathcal{V}_k$ be the admissible control $\forall k=1,...,K$. A vector $(v_1^*,v_2^*,...,v_K^*)$ is said to be a time consistent optimal strategy, any fixed initial state $(t,x,y)\in[0,T]\times\mathbb{R}\times\mathbb{R}$ and a fixed number $h>0$, a new vector $(v_1^h,v_2^h,...,v_K^h)$ defined by:
\[
    v^h(\tau,x,y)=\begin{cases}
        v_k(\tau,x,y) & \textnormal{for }(\tau,x,y)\in[t,t+h)\times\mathbb{R}\times\mathbb{R} \\
        v_k^*(\tau,x,y)& \textnormal{for }(\tau,x,y)\in[t+h,T]\times\mathbb{R}\times\mathbb{R}
    \end{cases}
\]
satisfies:
\[
    \lim\limits_{h\rightarrow 0}\inf \frac{J_k(t,x,y;v_1^*,v_2^*,...v_k^*,...,v_K^*)-J_k(t,x,y;v_1^*,v_2^*,...v_k^h,...,v_K^*)}{h}\geq 0,
\]
for $k=1, 2, \ldots, K$. In addition, the optimal value function is defined as:
\[
    u_k(t,x,y) = J_k(t,x,y;v_1^*,v_2^*,...,v_K^*).
\]
\end{definition}
\begin{remark}
The time consistent solution satisfies the conditions of a subgame perfect Nash equilibrium \citep{van-staden-2018}. We remark that in this paper, we use the term \textit{optimal} control instead of \textit{equilibrium} control in order to be consistent with other numerical approaches in the literature \citep{van-staden-2018,wang-2011,bjork-2017} .
\end{remark}

\section{Backward Stochastic Differential Equation (BSDE) Representation of the HJB Equation} \label{sec:FBSDE}

In this section, we present the optimal trade execution problem using an auxiliary HJB formulation and its BSDE. The HJB formulation has been shown to have a viscosity solution \citep{aivaliotis-2018}. We also show that when there is no price impact, the formulation reduces to the optimal investment problem. We reformulate the HJB to the BSDE in order to solve the optimal trade execution problem as a stochastic optimization problem which is more suitable for applications of neural networks.

\subsection{HJB Formulation}

We derive the general HJB equation for a performance conscious agent $k$. We first look at the HJB equation for a single agent with a single asset. To simplify the exposition, we drop the dependence on $t$ when it is clear. We define the value function such that:
\[
    u(S,b,\alpha,t;\phi) = \sup_{v}\{J(X(T);v,\phi)\}.
\]
The value function $u$ can be solved by reformulating the problem as an auxiliary problem \citep{wang-2011}. However, the auxiliary problems do not permit a viscosity solution. To overcome this, we can rewrite the variance term such that \citep{aivaliotis-2018}:
\[
    Var(X) = \mathbb{E}[X]^2 - \sup_{\psi\in\mathbb{R}}\{-\psi^2-2\psi\mathbb{E}[X]\}.
\]
Then the original mean-variance criteria becomes:
\begin{equation}
    u(S,b,\alpha,t) = \sup_{\psi\in\mathbb{R}}\{U(S,b,\alpha,t) - \frac{\gamma}{2}\psi^2\},
    \label{eq:dual_mv}
\end{equation}
where $U(S,b,\alpha,t)$ is the value function of the Markovian control problem given by \citep{aivaliotis-2018}:
\begin{equation}
    U(S,b,\alpha,t) = \sup_{v}\{(1-\gamma\psi)X(T) - \frac{\gamma}{2}(X(T))^2\}.
    \label{eq:markov_control}
\end{equation}
Note that $\psi^* = -\mathbb{E}^{v^*}[X(T)]$ solves \eqref{eq:dual_mv}. Let $\mathcal{L}U = (\mu-r)S\frac{\partial U}{\partial S} + \frac{(\sigma S)^2}{2}\frac{\partial^2 U}{\partial S^2}$. Then the HJB equation that solves \eqref{eq:markov_control} is given by:
\begin{equation}
    \frac{\partial U}{\partial t} + \mathcal{L}U + r\frac{\partial U}{\partial b} + \sup_{v\in \mathcal{V}}\{g(v)S\frac{\partial U}{\partial S} + v\frac{\partial U}{\partial \alpha}-vSf(v)\frac{\partial U}{\partial b}\} = 0,
\label{eq:hjb_markov}
\end{equation}
with terminal condition
\[
    U(T) = (1-\gamma\psi)X(T) - \frac{\gamma}{2}(X(T))^2.
\]
We extend this formulation to a performance conscious agent $k$ \citep{guan-2022} by replacing $X(T)$ with $\hat{X}(T) = X(T) - \phi\bar{X}(T)$.  We formally extend \eqref{eq:hjb_markov} to $K$ agents and $d$ assets in the following Lemma:
\begin{lemma}\label{lem:hjb}
Let the generator $\mathcal{L}U_k$ be defined as:
\[
    \mathcal{L}U_k = \sum_{i=1}^{d}(\mu_i-r)S_i\frac{\partial U}{\partial S_i} + \frac{1}{2}\sum_{i=1}^{d}\sum_{j=1}^{d}\rho_{i,j}\sigma_i\sigma_j S_iS_j\frac{\partial^2 U}{\partial S_i^2}.
\]
For each $k$ agent with assets $i=1, \ldots,d$, the extended HJB equation is given by:
\begin{align}
    \frac{\partial U_k}{\partial t} + \mathcal{L}_kU_k + rb_k\frac{\partial U_k}{\partial b_k} &+ \sup_{v\in \mathcal{V}}\{\sum_{i=1}^{d}g(\sum_{k=1}^{K} v_{i,k})S_i\frac{\partial U_k}{\partial S_i}\nonumber \\
    &+ \sum_{i=1}^{d}v_{i,k}\frac{\partial U_k}{\partial \alpha_{i,k}}-v_k\sum_{i=1}^{d}(S_if(v_{i,k}))\frac{\partial U_k}{\partial b_k}\} = 0,
    \label{eq:HJB_extended_agent}
\end{align}
with terminal condition
\[
    U_k(S_{1},..., S_d,b_k,\alpha_{1,k},...,\alpha_{d,k},T) = (1-\gamma\psi)X_k(T) - \frac{\gamma}{2}X_k(T)^2.
\]
\end{lemma}
\begin{proof}
Let $(S_1,...,S_d,b_k,\alpha_{1,k},...,\alpha_{d,k},t)$ be denoted as $(S_i,b_k,\alpha_{i,k},t)$. Let $\Delta t$ be the timestep size such that $\Delta t \rightarrow 0$, then $\Delta S_i = S_i(t+\Delta t) - S_i(t), \Delta b_k = b_k(t+\Delta t) - b_k(t)$, and $\Delta \alpha_{i,k} = \alpha_{i,k}(t+\Delta t) - \alpha_{i,k}(t)$. Then by the optimality of the control and law of iterated expectation for $i=1,...,d$ and for each $k$ the value function of the auxiliary problem is given by:
\begin{equation}
    U_{k} = \sup_{v\in\mathcal{V}}\mathbb{E}[U_{k}(S_i+\Delta S_i,b_k+\Delta b_k,\alpha_{i,k}+\Delta \alpha_{i,k},t+\Delta t)].
    \label{eq:markov}
\end{equation}
Let $\Delta U_{k} = U_{k}(S_i+\Delta S_i,b_k+\Delta b_k,\alpha_{i,k}+\Delta \alpha_{i,k},t+\Delta t) - U_{k}(S_i,b_k,\alpha_{i,k},t)$, then \eqref{eq:markov} can be written as
\begin{equation}
    \sup_{v_k\in\mathcal{V}}\mathbb{E}[\Delta U_{k}] = 0.
    \label{eq:dU}
\end{equation}
Applying Ito's Lemma on $\Delta U_{k}$, we get:
\[
    \mathbb{E}[\Delta U_k] = \mathbb{E} \left[\frac{\partial U}{\partial t}\Delta t + \sum_{i=1}^{d}\frac{\partial U}{\partial S_i}\Delta S_i + \frac{1}{2}\sum_{i=1}^{d}\frac{\partial^2 U}{\partial S_i^2}(\Delta S_i)^2 + \frac{\partial U}{\partial b_k}\Delta b_k + \sum_{i=1}^{d}\frac{\partial U}{\partial \alpha_{i,k}}\Delta \alpha_{i,k}\right],
\]
where
\begin{align*}
    \mathbb{E}[\Delta S_i] &= (\mu_i-r+\sum_{k=1}^{K}g_{i,k}(v_{i,k}))S_i\Delta t,\\
    \mathbb{E}[\Delta b_k] &= (rb_k -\sum_{i=1}^{d}v_{i,k}S_if_{i,k}(v_{i,k}))\Delta t,\\
    \mathbb{E}[\Delta \alpha_{i,k}] &= v_{i,k}\Delta t.
\end{align*}
The expectations of $\Delta S_i, \Delta b_k$ and $\Delta \alpha_{i,k}$ are straight forward. The expectation of $(\Delta S_i)^2$ requires more thought. Applying Ito's lemma on $(\Delta S_i)^2$, the expected value is given by:
\[
    \mathbb{E}[(\Delta S_i)^2] = \sum_{j=1}^{d}\rho_{i,j}\sigma_i\sigma_j S_iS_j\Delta t.
\]
Let $\mathcal{V}:=[v_{\min},v_{\max}]$ be the space of bounded trade rates $v$, such that there exists an optimal control $v_{i,k}^*$. Expanding \eqref{eq:dU} and evaluating $\mathbb{E}[\Delta S_i], \mathbb{E}[(\Delta S_i)^2], \mathbb{E}[\Delta b_k], \mathbb{E}[\Delta \alpha_{i,k}]$ gives us:
\begin{align*}
\left[\frac{\partial U_k}{\partial t} + \mathcal{L}U_k + rb_k\frac{\partial U_k}{\partial b_k}\right.
&+ \sup_{v\in\mathcal{V}} \left\{\sum_{i=1}^{d}\sum_{k=1}^{K}g_{i,k}( v_{i,k})S_i\frac{\partial U_k}{\partial S_i}\right.\\
&\left.\left. + \sum_{i=1}^{d}v_{i,k}\frac{\partial U_k}{\partial \alpha_{i,k}}-\sum_{i=1}^{d}(v_{i,k}S_if_{i,k}(v_{i,k}))\frac{\partial U_k}{\partial b_k}\right\}\right]\Delta t = 0.
\end{align*}
Diving out the $\Delta t$ and bringing the supremum inside the expectation gives us the extended HJB equation with terminal condition
\[
    U_k(S_1,...,S_d,b_k,\alpha_{1,k},...,\alpha_{d,k},T) = (1-\gamma\psi)X_k(T) - \frac{\gamma}{2}X_k(T)^2.
\]
\end{proof}
\begin{remark}
The $K$ HJB equations of Lemma \ref{lem:hjb} are coupled through the optimal control.
\end{remark}
\begin{remark}\textit{Special Case: $\kappa_\tau=0, \kappa_p=0$}\\
For the special case of $\kappa_\tau=0, \kappa_p=0$, i.e. there is no price impact, the optimal trade execution reduces to the optimal asset allocation problem. The analytical solution to the optimal asset allocation problem for time consistent mean-variance problem is known and has been solved in the linear-quadratic (LQ) sense \citep{zhou-2000,li-2000}. 
\end{remark}

\subsection{BSDE Formulation}

A stochastic approach has an advantage in that the HJB can be extended to include multiple assets and multiple agents. This is computationally infeasible with the PDE approach even with the dimensional reduction techniques \citep{forsyth-2011}. This is because each PDE for variance and wealth requires a three dimensional grid even with one asset. Though dimensional reduction can be used, this is does not apply for the general performance conscious agent due to the averaging term. We drop the explicit dependence on $t$ to keep the notation concise. To solve \eqref{eq:HJB_extended_agent}, we formulate the BSDE for $U_k$ using the following lemma.

\begin{lemma}
If there exists the optimal control $v_{i,k}^*$, the BSDE that satisfies viscosity solution for \eqref{eq:HJB_extended_agent} is given by:
\begin{equation}
    dU_k = \left(\frac{\partial U_k}{\partial t} + \mathcal{L}_kU_k\right) dt + \sum\limits_{i=1}^{d}\sigma_i\frac{\partial U_k}{\partial S_i} dW_i,
    \label{eq:BSDE_U}
\end{equation}
where
\begin{align*}
    \frac{\partial U_k}{\partial t} + \mathcal{L}_kU_k &= -rb_k\frac{\partial U_k}{\partial b_k} + \sum\limits_{i=1}^{d}v_{i,k}^*f_{i,k}(v_{i,k}^*)\frac{\partial U_k}{\partial b_k} \\
    &- \sum\limits_{i=1}^{d} (v_{i,k}^*)\frac{\partial U_k}{\partial \alpha_{i,k}} -\sum\limits_{i=1}^{d}\sum_{k=1}^{K}(g_{i,k}(v_{i,k}^*))S_{i}\frac{\partial U_k}{\partial S_{i}}.
\end{align*}
\end{lemma}
\begin{proof}
The proof is from the direct application of the BSDE formulation given in \citep{Pham-2014}. Let $x_{i,k} = \{S_i,b_k,\alpha_{i,k}\}$, for each agent $k$. Given the optimal control $v_{i,k}^*$ and a semi-linear PDE given by:
\[
    \frac{\partial U_k}{\partial t} + \mathcal{L}U_k + \sum_{i=1}^{d}F(x_{i,k},v_{i,k}^*,U_k,\frac{\partial U_k}{\partial x_{i,k}},t) = 0,
\]
with terminal condition $U_k(T) = \sum_{i=1}^{d}G(x_{i,k})$. The corresponding BSDE is given by:
\[
    dU_k = -\sum_{i=1}^{d}F(x_{i,k},v_{i,k}^*,U_k,\frac{\partial U_k}{\partial x_{i,k}},t)dt + \sum_{i=1}^{d}\sigma_i \frac{\partial U_{i,k}}{\partial x_{i,k}}dW_i,
\]
from the semi-linear PDE we can express $-F(x_{i,k},v_{i,k}^*,U_k,\frac{\partial U_k}{\partial x_{i,k}},t)$ as:
\[
    -\sum_{i=1}^{d}F(x_{i,k},v_{i,k}^*,U_k,\frac{\partial U_k}{\partial x_{i,k}},t) = \frac{\partial U_k}{\partial t} + \mathcal{L}U_k.
\]
From \eqref{eq:HJB_extended_agent} for each $k$ agents:
\begin{align*}
    -\sum_{i=1}^{d}F(x_{i,k},v_{i,k}^*,U_k,\frac{\partial U_k}{\partial x_{i,k}},t) &= -rb_k\frac{\partial U_k}{\partial b_k} + \sum\limits_{i=1}^{d}v_{i,k}^*f_{i,k}(v_{i,k}^*)\frac{\partial U_k}{\partial b_k} \\
    &- \sum\limits_{i=1}^{d} (v_{i,k}^*)\frac{\partial U_k}{\partial \alpha_{i,k}} -\sum\limits_{i=1}^{d}(g_{i,k}(\sum_{k=1}^{K}v_{i,k}^*))S_{i}\frac{\partial U_k}{\partial S_{i}}.
\end{align*}
\end{proof}
The significance of the system of BSDEs is two fold. One, the system of BSDEs does not require us to compute the Hessian of the extended HJB equation. This a big advantage as it allows us to avoid computing the second derivative on $U$ since computation of the second derivative can be very costly for high dimensions. This reduces the computational complexity of the problem. The second significance of the system of BSDEs is that it allows us to solve \eqref{eq:BSDE_U} in a least squares sense which is a recent technique used in high-dimensional option pricing \citep{han-2016,e-2017,hure-2020,chen-2019,na-2023}. This allows us to avoid using numerical solutions to the PDE and overcome the issue with dimensionality. Note, we can replace $X$ in \eqref{eq:dual_mv} and \eqref{eq:markov_control} with $\hat{X}$ for the case where the agent is performance conscious, i.e. $\phi > 0$. 

\subsection{Discretization of BSDE}

We use the Euler's method to simulate the forward SDE \eqref{eq:SDE} through time. Let $m=1,..,M$ be the indices of simulation paths and $n=0,...,N-1$ be the indices of the discrete time-steps from $0$ to $T$. Let $\Delta t = T/N$ be the timestep size and let $t^n = n\Delta t$. We discretize $dW_i^{t^n}$ as $(\Delta W_i)_m^n = \sum\limits_{j=1}^{d}L_{i,j}(Z_j)_m^n\sqrt{\Delta t}$ and \eqref{eq:SDE} as:
\begin{align}
    (S_i)_m^0 &= (s_i)_m^0, \nonumber\ &&\\
    (S_i)_m^{n+1} &= (1+(\mu_i - r + \sum_{k=1}^{K}g_{i,k}((v_{i,k})_m^n))\Delta t)(S_i)_m^n + \sigma_i(S_i)_m^n(\Delta W_i)_m^n,\ &&
    \label{eq:sdedis}
\end{align}
for $i=1, \ldots, d$. Next we discretize the position in the bank account, $b(t)$, using \eqref{eq:bank}:
\begin{equation}
    (b_k)_{m}^{n+1} = (1+r\Delta t)(b_k)_{m}^{n} -  \sum\limits_{i=1}^{d}(v_{i,k})_{m}^{n}(S_{i})_{m}^{n}f((v_{i,k})_{m}^{n})\Delta t.
    \label{eq:bank_discrete}
\end{equation}
The number of shares in stocks, $\alpha(t)$, is updated using \eqref{eq:trade_flow_agent}. Thus the discretized dynamics of $\alpha(t)$ is given by:
\begin{equation}
    (\alpha_{i,k})_{m}^{n+1} = (\alpha_{i,k})_{m}^{n}+(v_{i,k})_{m}^{n}\Delta t.
    \label{eq:alpha_discrete}
\end{equation}
The total wealth for each agent $k$ at each timestep can be calculated as:
\[
    (X_k)^n_m = (b_k)^n_m + \sum_{i=1}^{d}(\alpha_{i,k})^n_m(S_i)^n_m.
\]
When $\phi\neq 0$, we need to account for the cash account at each timestep $n$ by:
\[
    (\hat{X}_k)^{n}_{m} = (X_k)^n_m - \phi(\bar{X}_k)^n_m.
\]
Let $\hat{\mathcal{V}}$ be the discretization of $\mathcal{V}$ into $C\in\mathbb{Z}^{+}$ discrete points. Equations \eqref{eq:sdedis},  \eqref{eq:bank_discrete}, and \eqref{eq:alpha_discrete} are simulated forward in time over a subset of discretized control values $\hat{\mathcal{V}}\subset \mathcal{V}$.  The discrete set of control points  does not need to be too dense and we restrict it to the discrete set shown in \citep{forsyth-2011}.

To discretize the system of BSDEs, we first let $\{D_t, D_{S}, D_{b}, D_{\alpha}\}$ be the approximations of the partial derivatives with respect to $\{t,S,b,\alpha\}$,  respectively, defined as follows:
\begin{align*}
    D_tU &= \frac{U^{n+1} - U^{n}}{t^{n+1} - t^{n}},\\
    D_SU &= \frac{U^{n}(S +\Delta S) - U^{n}(S)}{\Delta S},\\
    D_bU &= \frac{U^{n}(b +\Delta b) - U^{n}(b)}{\Delta b},\\
    D_\alpha U &= \frac{U^{n}(\alpha +\Delta \alpha) - U^{n}(\alpha)}{\Delta \alpha}
\end{align*}
Then we can discretize \eqref{eq:BSDE_U} as:
\begin{align*}
    (U_k)^{n+1}_m - (U_k)^n_m &= \left(D_t( U_k)^{n}_m + (\mathcal{L}_{k}U_k)^n_m\right) \Delta t + \sum\limits_{i=1}^{d}\sigma_iD_{S_i}(U_{k})^n_m \Delta (W)_m^n \\
    &= (\Delta U_k)^n_m,
\end{align*}
where
\begin{align*}
    D_t (U_k)_{m}^{n} &+ (\mathcal{L}_kU_k)_{m}^{n} = -r((b_k)_{m}^{n} - \phi \bar{b}_{m}^n)D_b(U_k)_{m}^{n} \\
    &+ \sum\limits_{i=1}^{d}(v_{i,k})_m^nf((v_{i,k})_m^n)(S_{i})_m^n D_b (U_k)_m^n \\
    &- \sum\limits_{i=1}^{d} ((v_{i,k})_m^n- \phi (\bar{v}_{i})_m^n)D_\alpha (U_k)_{m}^{n} -\sum\limits_{i=1}^{d}(g(\sum_{k=1}^{K}v_{i,k})_m^n))(S_{i})_m^nD_{S}(U_k)_{m}^{n}.
\end{align*}
The recursive system of equations to solve for $(U_k)_m^n$ at time $n$ from $n+1$ is then given by:
\begin{align*}
    (U_k)_{m}^{n} &= (U_k)_{m}^{n+1} - (\Delta U_k)_{m}^{n}.
\end{align*}
We can approximate $U$ at timestep $n$ by:
\[
    (U_k)_{m}^{n} = (U_k)_{m}^{n+1} - (\Delta U_k)_{m}^{n},
\]
and enforcing time-consistency through each discrete timestep the approximation of $U$ is given by:
\begin{equation}
    \label{eq:BSDE_discrete_U}
    (U_k)_{m}^{n} = \mathbb{E}_{v^{*n+1}}[(U_k)_{m}^{n+1}] - (\Delta U_k)_{m}^{n}.
\end{equation}
Note that the BSDE \eqref{eq:BSDE_discrete_U} is solved backward in time.

\subsubsection{Least squares formulation}

To solve the BSDE \eqref{eq:BSDE_discrete_U}, we approximate $(U_k)^n$ by $(\hat{U}_k)^n$ and reformulate the problem in a least squares formulation. Consider the $n^{th}$ timestep from \eqref{eq:BSDE_discrete_U}, the pointwise residual $(R_U)^n_m$ is given by:
\[
    (R_U)^n_m = (U_k)^{n+1}_m - ((\hat{U}_k)^{n}_m-\left(D_t (\hat{U}_k)^{n}_{m} + (\mathcal{L}\hat{U}_k)^{n}_m\right) \Delta t + \sum\limits_{i=1}^{d}\sigma_iD_{S_i}(\hat{U}_k)^{n}_m \Delta (W)^n_m).
\]
The objective becomes the minimization of the residuals such that $(R_U)^n_m \rightarrow 0$. Since our $K$ BSDE equations are coupled by the optimal control and solved backwards in time, we satisfy the optimality condition which solves the multi-agent game as in Definition \ref{def:optimal_control}.

This gives us the criteria we need to form the following least squares problem:
\begin{equation}
    L_{MSE} = \mathbb{E}[||(R_U)^n||_2^2].
    \label{eq:loss_U}
\end{equation}
Minimizing \eqref{eq:loss_U} gives us a good approximation to $(U_k)^n$ . The optimal control $v^{*n}$ can then be determined by:
\[
    (v_{i,k}^*)^{n} = \textnormal{arg}\max_{v\in\hat{\mathcal{V}}}\mathbb{E}_v[(\hat{U}_{k})^n].
\]
\begin{remark}\textit{Computing gain and risk}\\
    Given the optimal control $v_{i,k}^{*n}$ that maximizes  \eqref{eq:cost_mv_opt}, we can compute the expected gain and risk of the strategy by running a Monte Carlo simulation to compute the terminal wealth. The Monte Carlo follows the algorithm as shown in \citep{tse-2012}.
\end{remark}

\section{Neural network approach to the optimal trade execution problem}

In this section, we propose a neural network approach to solve the system of least squares formulation \eqref{eq:loss_U}. The method follows a general framework that has been used to solve high-dimensional American option pricing problems and has shown to produce good results \citep{han-2016,e-2017,hure-2020,chen-2019,na-2023}. This framework has also been used to solve more general HJB equations as shown in \citep{hure-2020}.

\subsection{Neural network formulation}
We model the approximations $\hat{U}_m^n$ using a deep residual U-net with a self attention layer. U-net is a convolutional neural network architecture with great success in image processing and was first introduced for the application of biomedical image segmentation \citep{ronneberger-2015}. We chose this architecture to learn the relationship between the discretized trade rates and assets. This is because U-net has been shown to use and augment existing features more effectively \citep{ronneberger-2015}. We also explore the use of self-attention layer to increase the efficiency of the residual deep network by embedding temporal relationship between the previous estimate $(U)^{n+1}_m$ and the current states given by \{$(S)^n_m, (b)^n_m, (\alpha)^n_m$, $(v)^n_m$\}. The use of a self-attention layer allows our network to determine which features are the most relevant in predicting  the output \citep{vaswani-2017}. This allows us to avoid using recurrent network structures which have a chance to exhibit vanishing gradient properties due to the activation functions used in their gate structures \citep{vaswani-2017}.

\subsubsection{Self-attention residual network}

The standard attention network was initially used to model predictive sequential language models \citep{vaswani-2017}. They were initially implemented to learn the importance of features in predicting a language sequence. The self-attention architecture was introduced as a method to embed information inputs to create latent feature space for more efficient training \citep{vaswani-2017}. In this paper, we use self-attention layer to overcome the limitation of deep residual network methods. Deep residual networks need to save weights at every timestep \citep{han-2016,e-2017,chen-2019}; however, by adding the self-attention layer allows us to capture the dependencies between features.

The attention mechanism relies on a query $Q$, and key-value pair $\{K,\tilde{V}\} $. We let $A(Q,K,\tilde{V})$ be the score function given by \citep{vaswani-2017}:
\[
    A(Q,K,V) = softmax\left(\frac{QK^\top}{\sqrt{l}}\right)\tilde{V},
\]
where the softmax function is a standard activation function \citep{goodfellow-2016} and $l$ is the length of the sequence. Let $h=1,...,H$ be the number of heads where each head is given by $\mathcal{H}_h = A(Q,K,\tilde{V})$. It was found that constructing multiple attention layers in a bagging approach allowed for better performance than a single attention mechanism \citep{vaswani-2017}. The multi-head attention mechanism is the concatenation of $h$ heads and is defined as \citep{vaswani-2017}:
\[
    \mathcal{M}_H = \mathcal{H}_1 \oplus ... \oplus\mathcal{H}_{H-1}\oplus\mathcal{H}_H,\ h=1,...,H,
\]
where $\oplus$ is used to signify concatenation. It has been shown that using a linear layer to represent $\{Q,K,\tilde{V}\}$ is beneficial to learning \citep{vaswani-2017}. The self-attention network is a special case of the multi-head attention mechanism where the inputs are passed through a linear layer to construct $\{Q,K,\tilde{V}\}$. More specifically, given an input $X$ and weights $W_Q$, $W_K$, $W_{\tilde{V}}$ we can define $Q = W_Q X$, $K = W_K X$, $\tilde{V} = W_{\tilde{V}} X$. 

The self-attention mechanism itself does not have much prediction power \citep{vaswani-2017}. To add prediction power to the self-attention layer, we couple it with a convolution neural network (CNN). The benefit of using a CNN is that it captures the high dimensional structure across the discretized space of trade rates and our input features \citep{goodfellow-2016}. More specifically, we use a U-net architecture coupled with a residual structure. The U-net we use has $4$ convolution layers and $4$ transposed convolution layers with a converging-diverging geometry as shown in Figure \ref{fig:unet} \citep{ronneberger-2015}. The convergent-divergent structure of the U-net allows for the network to capture the content in the data and allows for localization \citep{ronneberger-2015}. It also allows us to limit the effects of features that may not contribute to predictions. 
\begin{figure}
    \centering
    \includegraphics[width=0.95\textwidth]{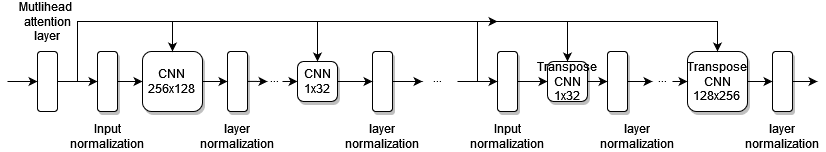}
    \caption{U-net with self attention architecture. It is composed of a multihead attention layer, convolution layers and decoding transpose convolution layers with layer normalization. The input data is added to the output of each layer creating a residual structure.}
    \label{fig:unet}
\end{figure}

In between each network layer, we normalize the input $X$. Normalization allows our network to learn more effectively and prevent cross correlation between network layers \citep{goodfellow-2016}. For learnable parameters $\varphi$ and $\eta$ to assist convergence of the deep residual network, we use group normalization, $G$, given by \citep{he-2016}:
\[
    G(X,\eta,\varphi) = \eta\frac{X-\mathbb{E}[X]}{\sqrt{Var(X) + \varepsilon}} + \varphi,
\]
where $\varepsilon\ll 1$ is some small number to aid in numerical stability. 
\begin{figure}
    \centering
    \includegraphics[width=0.8\textwidth]{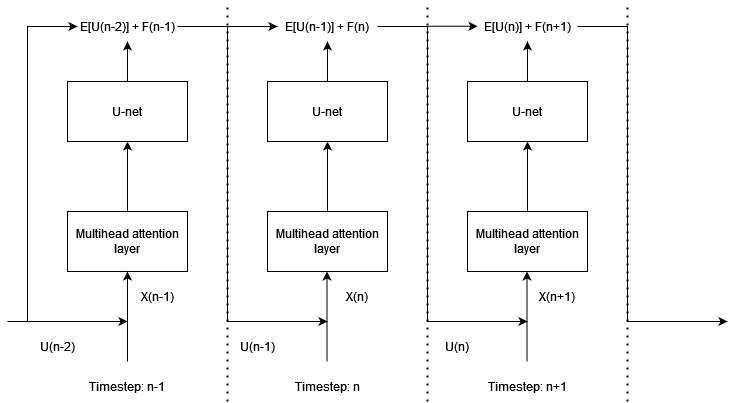}
    \caption{Data flow illustration of the residual U-net with self-attention rolled out through time.}
    \label{fig:selfattresunet}
\end{figure}

For a given input, $X^n$, at timestep $n$, the residual U-net for each timestep has the following form \citep{he-2016}:
\[
    Res_U^n(X^n) = \textnormal{U-net}(X^n) + W_{R} Res_U^{n-1}(X^{n-1}),
\]
where $W_R$ is the weight of a linear layer. The residual U-net with self-attention used in this paper is shown in Figure \ref{fig:selfattresunet}. The network becomes a deep network as we unroll each network through time. By adding $(U_k)^{n+1}$ to the output, it gives our network a residual network structure through time.

To solve \eqref{eq:BSDE_discrete_U}, we first approximate the function $U_k$ using the deep residual network with self attention. Given the set of weights $\Theta$, let $\hat{U}_k(\Theta)$ be the neural network approximation to $U_k$. The initial input $X$ is used to construct the self-attention mechanism $\mathcal{M}_H$, where each head is given by $\mathcal{H}_h = A(W_Q X,W_K X,W_{\tilde{V}} X)$. The output $X_{a}$ is then added to the original input and normalized. More precisely, $X_a$ is given by:
\[
    X_a = G(\mathcal{M}_H + W_a X).
\]
This new input is then passed through our residual U-net. This allows us to learn spatial relationship between our features and the discretized trade rates.

In our network, we use the swish activation function over the ReLU activation function. This is because the ReLU activation function may cause information loss when training over long sequences due to the negative values being zeroed out \citep{na-2023}. The output of CNN plus some weighted $X_a$ gives us the final output:
\[
    F(\Theta) = Res_U^n(X_a;).
\]
For timestep $n$ and sample $m$, the neural network approximation to $(U_k)^n_m$ is given by:
\[
    (\hat{U}_k)^n_m(\Theta) = \mathbb{E}[(U_k)^{n+1};(v_{i,k}^*)^{n+1}] + F(\Theta).
\]
Note that the number of time-steps $N$ represents the depth of the deep residual self-attention network. 

\subsection{Network Training and Prediction}

Let $(\hat{U}_k)^n_m(\Theta_k)$ be the deep self-attention residual network approximation of $(U_k)^n_m$ with parameters $\{\Theta_k\}$ for agent $k$. The parameters $\{\Theta_k\}$ are the weights and biases of the deep self-attention residual network. The input data of $(\hat{U}_k)^n_m(\Theta_k)$ is given by $(X_{input})^n_m = \{(S_i)^n_m,(b_k)^n_m,(\alpha_{i,k})^n_m,(v_{i,k})^n_m, (U_k)^{n+1}_m,n\Delta t,m,c\}$. We concatenate the output from the previous time-step into our input. This information is used by the self-attention mechanism to embed previous timestep information into the model. The last three inputs give positional data for time, simulation path, and control index on the discretized set of control points $\hat{\mathcal{V}}$. The label data used to train the network is given by the value function from the previous timestep and is given by:
\begin{align*}
    (Y_{label})^n_m &= (U_k)^{n+1}_m.
\end{align*}
This is because the BSDE is solved backward in time and we must satisfy the terminal condition \eqref{eq:markov_control}. We let
\[
f^n_m = \left(D_t(\hat{U}_k)^{n}_m + (\mathcal{L}_k\hat{U}_k)^{n}_m\right),
\]
where
\begin{align*}
    D_t(\hat{U}_k)_{m}^{n} &+ (\mathcal{L}_k\hat{U}_k)_{m}^{n} = -r((b_k)_{m}^{n} - \phi \bar{b}_{m}^n)D_b(\hat{U}_k)_{m}^{n} \\
    &+ \sum\limits_{i=1}^{d}(v_{i,k})_m^n f((v_{i,k})_m^n)(S_{i})_m^nD_b (\hat{U}_k)_m^n \\
    &- \sum\limits_{i=1}^{d} ((v_{i,k})_m^n- \phi (\bar{v}_{i})_m^n)D_\alpha (\hat{U}_k)_{m}^{n} -\sum\limits_{i=1}^{d}(g(\sum_{k=1}^{K}v_{i,k})_m^n))(S_{i})_m^nD_S (\hat{U}_k)_{m}^{n},
\end{align*}
and let
\[
    Z^n_m = \sum\limits_{i=1}^{d}\sigma_iD_S(\hat{U}_k)^{n}_m \Delta (W_i)^n_m.
\]
We rewrite the losses \eqref{eq:loss_U} such that 
\begin{align*}
    (L_U)^n(\Theta_k) &= \mathbb{E}[((R_U)^n(\Theta_k))^2] = \mathbb{E}[((Y_{label})^{n} - ((\hat{U}_k)^n(\Theta)-f^n\Delta t+Z^n))^2].
\end{align*}
Minimizing $L_U(\Theta_k)$ gives us the optimal set of parameters for the self-attention network $(\hat{U}_k)^n(\Theta_k)$ given by:
\[
    (\Theta_k^*)^n = \textnormal{arg}\min (L_U)^n(\Theta_k).
\]
We summarize the training of the proposed model in Algorithm \ref{alg:training}. 
The gradients are computed using autodifferentiation.

\begin{algorithm}[htbp]
\begin{minipage}{110mm}
    \caption{Training a deep residual network with self-attention, we denote the neural network as $NN$}
    \label{alg:training}
    \begin{algorithmic}
        \Require{$(X_{input})^n_m,(Y_{label})^n_m$, $(U_k)^N_m$,$NN$}
        \Let{$loss_k$}{[]}
        \For{$k = 1,...,K$}
            \Let{$loss_n$}{[]}
            \For{$n=N-1,...,0$}
                \Let{$(F(\Theta_k)$}{$NN((X_{input})^n_m)$}
                \Let{$(\hat{U}_k)^{n}_m(\Theta_k)$}{$\mathbb{E}[(U_k)^{n+1}_m]+F(\Theta_k)$}\Comment{compute $f^n_m, Z^n_m$}
                \Let{$loss_n$}{$(L_U)_k^n(\Theta_k)$} \Comment{accumulate losses}
            \EndFor
            \Let{$loss_k$}{[]} \Comment{accumulate losses}
        \EndFor
    \Let{$\Theta^*$}{$\arg\min\{loss_k\}$}
    \end{algorithmic}
\end{minipage}
\end{algorithm}

To compute $(v_{i,k}^*)^{n}$ for agent $k$ we run an inference of the network. We perform a linear search for the optimal expected trade rate for each asset $i$ over the grid of $\hat{\mathcal{V}}$. We summarize the inference step of the proposed model in Algorithm \ref{alg:prediction}.

\begin{algorithm}[htbp]
\begin{minipage}{110mm}
    \caption{Inference of value function from a deep residual network with self-attention, we denote the neural network as $NN$.}
    \label{alg:prediction}
    \begin{algorithmic}
        \Require{$X^n_m$, $(U_k)^N_m$,$NN$}
        \For{$k = 1,...,K$}
            \For{$n=N-1,...,0$}
                \Let{$(F(\Theta_k^*)$}{$NN((X)^n_m)$}
                \Let{$(\hat{U}_q)^{n}_m$}{$\mathbb{E}[(U_{k})^{n+1}_m;v_{i,k}^*]+F(\Theta_k^*)$}
                \Let{$(v_{i,k}^*)^n$}{$\arg\min\limits_{v\in\hat{\mathcal{V}}}\{\mathbb{E}_v\{\hat{U}_k^n\}\}$}
            \EndFor
        \EndFor
    \end{algorithmic}
\end{minipage}
\end{algorithm}

Note, due to the non-uniqueness of the optimal control, we train multiple networks for each agent and perform bagging to ensure better results. This allows our neural network to explore multiple paths. Optimization is done using the Adam algorithm \citep{kingma-2014} with a learning rate of $1\times10^{-4}$. 

\subsection{Data Generation}

To construct the training data used by the deep residual self-attention network, we need to perform forward simulations for the dynamics of the underlying asset $(S_i)^n_m$, cash position $(b_k)^n_m$, position in stocks $(\alpha_{i,k})^n_m$, and the control $(v_{i,k})^n_m$. At $n = 0$, we initialize the variables according the initial conditions given by \eqref{eq:bc_0} for the variables $\{(S_i)^0_m,(b_k)^0_m,(\alpha_{i,k})^0_m\}$. The trade rate $(v_{i,k})^n_m$ is constructed as a uniform grid of discretized points over $c = 1,...,C$ points. The trade rate is given by:
\[
    v_{i,k} = \{v_{min}/T,....,v_{max}/T\},
\]
where $v_{i,k}$ is in units of  $1/T$. For $n=1,...,N$, the dynamics of \eqref{eq:sdedis} is simulated using Euler's method. The dynamics of $\{(b_k)^n_m,(\alpha_{i,k})^n_m\}$ follow the dynamics given by \eqref{eq:bank_discrete} and \eqref{eq:alpha_discrete}. To ensure that all assets are liquidated at time $T$, we define the trade rate, $v_T$, at the instant before maturity as \citep{forsyth-2011}:
\[
    v_T = \frac{-(\alpha_{i,k})^N_m}{dT},\ \textnormal{where}\ dT \ll dt.
\]
The cash position is updated one final time at $t = T-dT$ so as to penalize all positions that have not been liquidated at the final trading step. This is done at the end of the forward simulation in order to compute the value function at terminal time $t=T$.

\section{Numerical Results}

In this section, we outline the numerical experiments to validate the proposed method. We demonstrate the effects of different parameters on the proposed method and the numerical studies on different strategies used in optimal trade execution. 

In Experiment $1$, we look at a special case when the trade impact factors $\kappa_p = 0, \kappa_\tau = 0$ to validate the proposed method. In Experiment $2$  we observe at the effects of different parameters on $\alpha^*$. In Experiment $3$, we study the case of $2$ agents where one agent has a large permanent price impact trading in an environment while the other agent has no impact. In this setting, we also vary $\phi=0.0, 0.3, 0.7$ and consider the cases when the economic conditions are good $(\mu-r > 0)$ and bad $(\mu-r < 0)$. In Experiment $4$, we study the neural network solution to the optimal execution problem with $2$ agents on different strategies in different market conditions. We also look at the effects of correlation when agents are trading $d=2$ assets. In Experiment $5$, we study the neural network solution with $10$ agents and $2$ assets. We divide the $10$ agents into buyers and sellers and demonstrate the performance of each group of liquidating agents and compare them to a strategy where they just hold the asset from $[0,T]$.

All experiments are done using a 6GB-NVIDIA GTX 2060 GPU, a 6 core AMD Ryzen 5 3600X processor and 16GB of RAM.

\subsection{Experiment 1: Comparison of Residual Attention Network to Existing solutions}

In this experiment, we compare the approximation of the neural network solution to the results of existing work \citep{forsyth-2011, aivaliotis-2018, guan-2022, zhou-2000}.  This experiment verifies that the outputs of our neural network solution coincides with known solutions wherever they are available. We fix the parameter $\gamma = 6, \phi = 0.0$ and assume no price impact, i.e. $\kappa_p =0, \kappa_{\tau} = 0$ and $\kappa_s = 0$.

\vskip2mm
\noindent\textit{Single agent, single asset}\vskip2mm

In the case of a single agent with one asset, we compare the results to the FDM solution \citep{forsyth-2011, aivialotis-2014}. First, we look at the performance of the proposed method for different $N$. This is done to determine the appropriate number of timesteps we need for an accurate solution to the problem. We fix $M = 3000$, $T=1$ year and run the neural network until $L_U(\Theta) < 5\times 10^{-6}$ and summarize the results in Table \ref{tab:convergence_N}. We also record the FDM solution for $N=10,100,200,300,400$ with relative error given by:
\begin{equation}
    Err_{rel} = \frac{|actual - approx.|}{|actual|} \times 100\%.
    \label{eq:rel_err}
\end{equation}
We treat the FDM solution as the actual solution and the NN solution is the approximation.

\begin{table}[htbp]
    \centering
    \begin{tabular}{c|c|c|c}\hline 
        $N$ & FDM & Proposed method  &Relative error\\ \hline 
        $10$ & $1.0565641$ & $1.051620$ &$0.4679\%$\\ \hline
        $100$ & $1.0565812$ & $1.056400$  &$0.017\%$\\ \hline 
        $200$ & $1.0565820$ & $1.056500$ &$0.0078\%$\\ \hline 
        $300$ & $1.0565823$ & $1.056309$ &$0.026\%$\\ \hline 
        $400$ & $1.0565824$ & $1.056879$ &$0.028\%$\\ \hline
    \end{tabular}
    \caption{Residual U-net with self-attention approximation for the solution to the optimal value $u(x_0,0)$ over different number of time-steps $N$. The FDM values are results from \citep{aivialotis-2014}. The relative errors are computed using \eqref{eq:rel_err} }
    \label{tab:convergence_N}
\end{table}

We see that the proposed method has the smallest relative error to the FDM solution at $N=200$. We found that a network depth $N>200$ tends to over fit the data, which is shown by the decrease in accuracy. 

Next, we show the sample size $M$ that minimizes the relative error. This is done to determine the appropriate batch size we need to accurately compute our results. We fix $N = 300$ and run the neural network until convergence. The results are summarized in Table \ref{tab:convergence_M}. In Table \ref{tab:convergence_M}, we record the FDM solution for the number of space grids $100,1000,2000,3000$. We observe that the proposed method has the smallest relative error to the FDM solution at $M=1000$.

\begin{table}[htbp]
    \centering
    \begin{tabular}{c|c|c|c}\hline 
        $M$ & FDM & Proposed method  &Relative Error\\ \hline 
        $100$ & $1.0553888$ & $1.053110$ &$0.216\%$\\ \hline 
        $1000$ & $1.0565197$ & $1.056694$ &$0.016\%$\\ \hline 
        $2000$ & $1.0565823$ & $1.056762$ &$0.017\%$\\ \hline 
        $3000$ & $1.0566032$ & $1.056879$ &$0.027\%$\\ \hline 
    \end{tabular}
    \caption{Residual U-net with self-attention approximation for the solution to the mean-variance criterion $u(x_0,0)$ over different number of simulations $M$. In the FDM solution $M$ represents the grid size used for the wealth \citep{aivialotis-2014}. The relative error is computed using \eqref{eq:rel_err}}
    \label{tab:convergence_M}
\end{table}

We also make a comparison to the FDM solution to the optimal trade execution for $\kappa_\tau = 2\times10^{-6}$, $T=1/250$, $\mu = 0$ and $r = 0$ \citep{forsyth-2011}. We compute the efficient frontier \citep{tse-2012} as shown in Figure \ref{fig:EF}. For an initial price of $s(0) = 100$ and $\alpha(0) = 1.0$, the expected gain from selling is $99.295$ with a standard deviation of $0.7469$ \citep{forsyth-2011}. Our neural network approximation gives an expected gain of $99.271$ with a standard deviation of $0.7402$, similar to the FDM results.

\begin{figure}[htbp]
    \centering
    \includegraphics[width=0.4\textwidth]{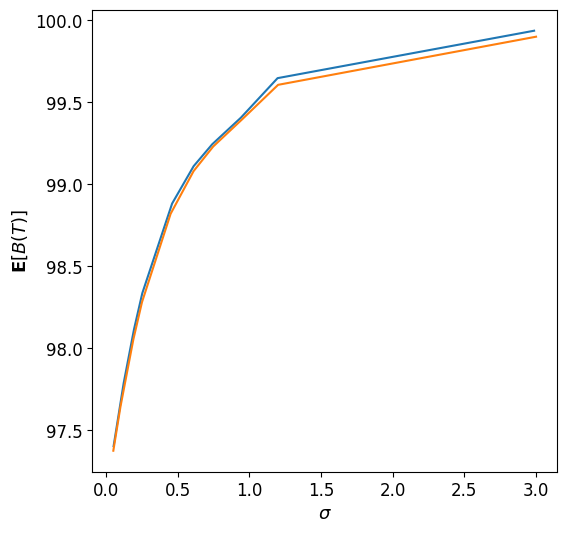}
    \caption{The efficient frontier interpolated from the neural network approximation (blue) and interpolated values from \citep{forsyth-2011} (orange). Interpolation was performed using a polynomial degree $4$.}
    \label{fig:EF}
\end{figure}

\vskip2mm
\noindent\textit{Single agent, $2$ assets}\vskip2mm

We compare the neural network solution with analytical solutions. For the special case of no price impact, an analytical solution exists for the case of a single agent with multiple assets \citep{zhou-2000}. This experiment shows that our outputs are consistent with multidimensional results. 

Let $A=[\mu_1-r,....,\mu_d-r]^\top, \underline{B} = [\sigma_1,...,\sigma_d]^\top$ be the vector of drift rates and the covariance matrix for each assets $i=1,...,d$. Note this framework assumes independence between the assets, and hence we represent the volatilities as a vector. The time-consistent mean variance optimal asset allocation for multiple assets is given by \citep{zhou-2000}:
\[
    \textnormal{Minimize}\ J(X(T); \mu, \lambda) \equiv \mathbb{E}[\mu X(T)^2 - \lambda X(T)],
\]
where parameter $\mu > 0$ and $-\infty < \lambda < \infty$. Let, $\mu = 1$, $\lambda = \gamma$ and let $\alpha^*(t)$ be the quantity of assets held by the agent. Then the optimal asset allocation for a single agent with multiple assets is given by \citep{zhou-2000}:
\begin{equation}
    \alpha^*(t) = [\underline{B}\underline{B}^\top]^{-1}\underline{A}^\top(\frac{\gamma}{2} e^{-r(T-t)}-\underline{X}(t)),
    \label{eq:xyz_solution}
\end{equation}
where $\underline{X}(t) = [X_1(t),...,X_d(t)]^\top.$ Our model can be used to compute $\alpha^*(t)$ by solving the equation $\alpha^*(t) = \alpha^*(t-dt) + v^*(t) dt$.
We  fix $N=100$, and $M = 1000$ and compute $\mathbb{E}[J(X(T));\alpha^*(t),t=0,\Delta t,...,N\Delta t]$ using Monte Carlo simulation \citep{zhou-2000}. Given $\{\alpha_{1}(0)=1.0, \alpha_{2}(0) = 0.5\}$, the optimal value for $d=2$  at $T=1/250,1/52,1/12,1$ years are summarized in Table \ref{tab:convergence_xyz}. We see that the highest relative error occurs at $T = 1$ with a relative error of $0.138\%$. 

\begin{table}[htbp]
    \centering
    \begin{tabular}{c|c|c|c}\hline 
        $T$ & Analytical Solution & Proposed method  &Relative Error\\ \hline 
        $1/250$ & $1.500340$ & $1.500390$ &$0.0033\%$\\ \hline
        $1/52$ & $1.501912$ & $1.501507$  &$0.027\%$\\ \hline 
        $1/12$ & $1.506883$ & $1.506248$ &$0.042\%$\\ \hline 
        $1$ & $1.583452$ & $1.585642$ &$0.138\%$\\ \hline 
    \end{tabular}
    \caption{The residual U-net with self attention approximation to the solution of the optimal value $u(s_0,b_0,\alpha_0,0)$ for different maturities $T$ in years. The analytical solution is given by \eqref{eq:xyz_solution}. We calculate the relative error using \eqref{eq:rel_err}.}
    \label{tab:convergence_xyz}
\end{table}

\vskip2mm
\noindent\textit{$2$ agents, single asset}\vskip2mm

We compare the neural network solution with analytical solutions. For the special case of no price impact, an analytical solution exists for multiple agent with one asset \citep{guan-2022}. 
The mean-variance optimal portfolio problem can be written as \citep{guan-2022}:
\[
    \textnormal{Minimize}\ J(X(T)) \equiv \mathbb{E}[X(T)] - \frac{\gamma}{2}Var(X(T)).
\]
The analytical solution to the time-consistent auxiliary problem is given by \citep{guan-2022}:
\begin{equation}
    \alpha_k^* = \frac{\mu_k}{\gamma_k\sigma_k^2} + \frac{\phi_k}{\sigma_k}\frac{1}{K}\sum_{k=1}^{K}\frac{\mu_k}{\gamma_k\sigma_k(1-\phi_k)}.
    \label{eq:gh_solution}
\end{equation}
We use Monte Carlo Simulation to compute $\mathbb{E}[J(X(T));\alpha^*(t),t=0,\Delta t,...,N\Delta t]$ \citep{guan-2022}. Given $\alpha_{k}(0)=1.0$ the optimal value for $K=2$ are summarized in Table \ref{tab:convergence_gh}. We present results for only one agent as the agents are assumed to be symmetric. We can see from Table \ref{tab:convergence_gh} that the relative errors are generally small with the highest relative error of $0.0351\%$ occurred at $T = 1/12$. 

\begin{table}[htbp]
    \centering
    \begin{tabular}{c|c|c|c}\hline 
        $T$ & Analytical Solution & Proposed method  &Relative Error\\ \hline 
        $1/250$ & $1.000221$ & $1.000216$ &$0.0005\%$\\ \hline
        $1/52$ & $1.001059$ & $1.001032$  &$0.002\%$\\ \hline 
        $1/12$ & $1.004577$ & $1.004224$ &$0.0351\%$\\ \hline 
        $1$ & $1.0565823$ & $1.056694$ &$0.0122\%$\\ \hline 
    \end{tabular}
    \caption{The residual U-net with self-attention approximation to the optimal value, $u(s_0,b_0,\alpha_0,0)$, for different maturities $T$ in years. The analytical solution is given by \eqref{eq:gh_solution}. We calculate the relative error using \eqref{eq:rel_err}.}
    \label{tab:convergence_gh}
\end{table}

To summarize, the proposed model performs well in both multi-agent and multi-asset cases.

\subsection{Experiment 2: Sensitivity of Parameters}

\begin{figure}[htbp]
    \centering
    \includegraphics[width=0.33\textwidth]{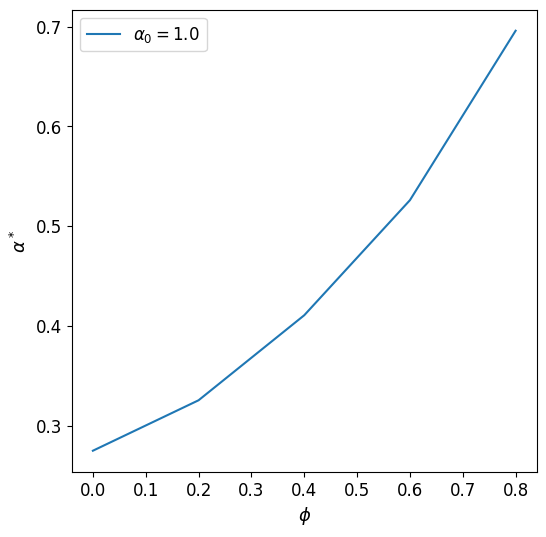}
    \hspace{0.5cm}
    \includegraphics[width=0.33\textwidth]{Figures/P_sensitivity_2K1D_nn_solution.png}
    \caption{The sensitivity of of $\alpha^*$ with respect to $\phi=[0.0, 0.2, 0.4, 0.6, 0.8]$ at $t=1/250$. The sensitivity of agent $1$ (left) and the agent $2$ (right)}
    \label{fig:phi_sensitivity}
\end{figure}

In this experiment, we study the effects of the relative performance factor $\phi$, the risk preference $\gamma$, the permanent price impact factor $\kappa_p$ and the temporary price impact factor $\kappa_t$ on the optimal portfolio $\alpha^*$. We vary each of these factors independently and observe the portfolio at different time-steps for $2$ agents with $1$ asset. For the sensitivity study we assume that agents want to liquidate long positions (sell only). We set $\alpha_0=1.0$, $s_0=1.0$, $b_0=0.0$, $\mu = 0.1$, $r = 0.05$, $\sigma = 0.2$ and $T=1/250$. 

To analyze the sensitivity of $\phi$, we let $\gamma = 6$, $\kappa_p = 0$ and $\kappa_\tau=0$. The results are shown in Figure \ref{fig:phi_sensitivity}. As we can see, as $\phi$ increases, $\alpha^*$ also increases. This is expected since the agent invests more into risky assets as they become more aware of other agents \citep{guan-2022}. Thus our neural network approach is able to capture this behaviour correctly.

\begin{figure}[htbp]
    \centering
    \includegraphics[width=0.33\textwidth]{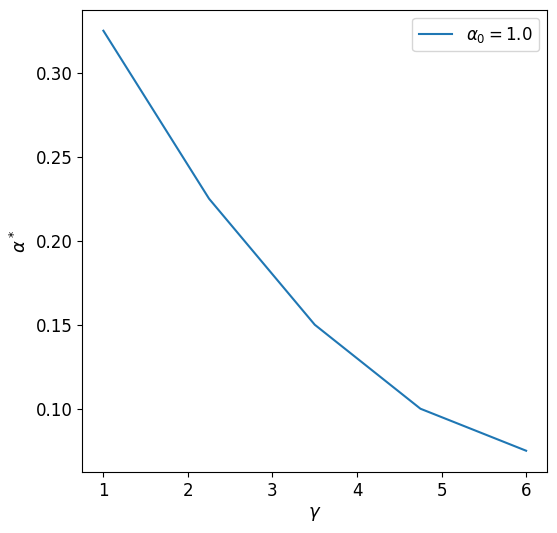}
    \hspace{0.5cm}
    \includegraphics[width=0.33\textwidth]{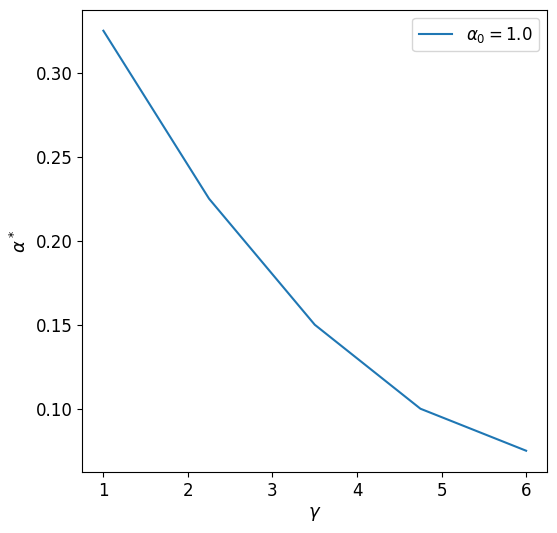}
    \caption{The sensitivity of of $\alpha^*$ with respect to $\gamma=[3.0, 3.75, 4.5, 5.25, 6.0]$ at $t=1/250$. The sensitivity of agent $1$ (left) and the agent $2$ (right).}
    \label{fig:gamma_sensitivity}
\end{figure}

Moreover, we also want to ensure that $\alpha^*$ decreases as $\gamma$ increases. This behaviour is an important behaviour to observe in agents as it is commonly known that as the risk aversion increases, the agent invests less in risky assets \citep{zhou-2000}. We fix $\phi = 0.0$, $\kappa_p = 0.0$, and $\kappa_\tau = 0.0$. 
Figure \ref{fig:gamma_sensitivity} shows that $\alpha^*$ indeed decreases as $\gamma$ increases, illustrating the effectiveness of the neural network approach. 

\begin{figure}[htbp]
    \centering
    \includegraphics[width=0.37\textwidth]{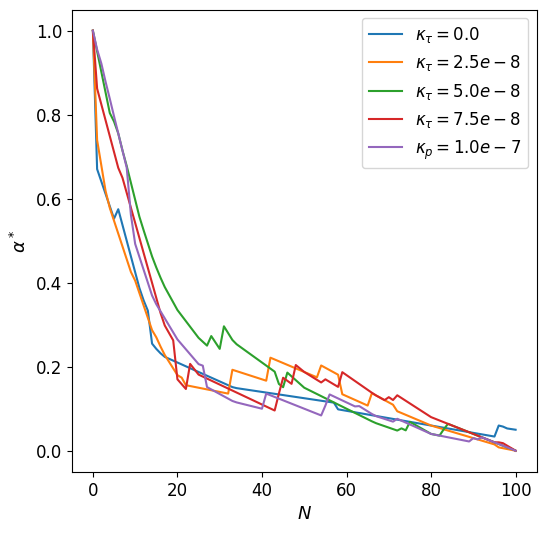}
    \hspace{0.3cm}
    \includegraphics[width=0.37\textwidth]{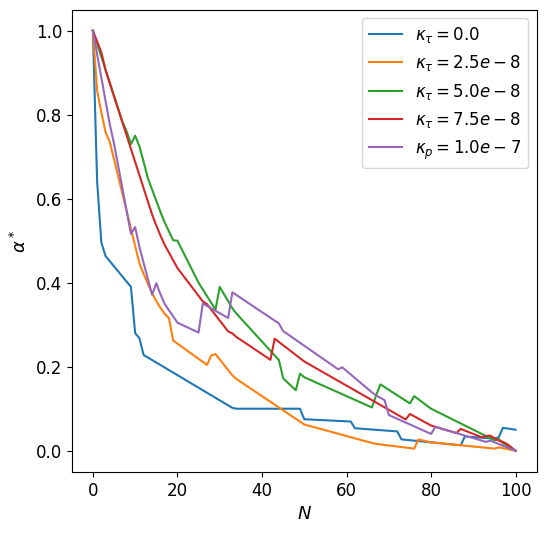}
    \includegraphics[width=0.37\textwidth]{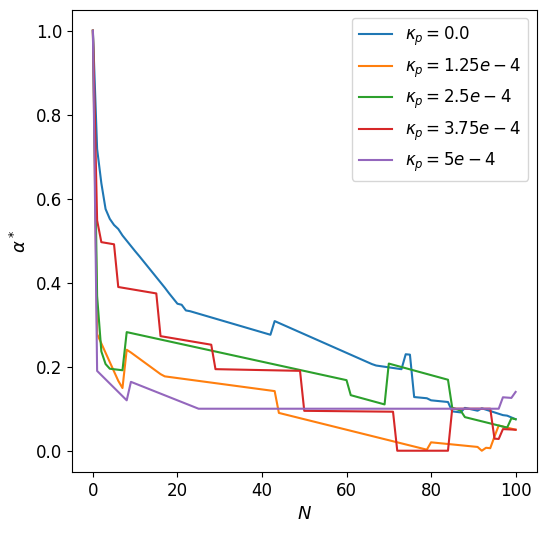}
    \hspace{0.3cm}
    \includegraphics[width=0.37\textwidth]{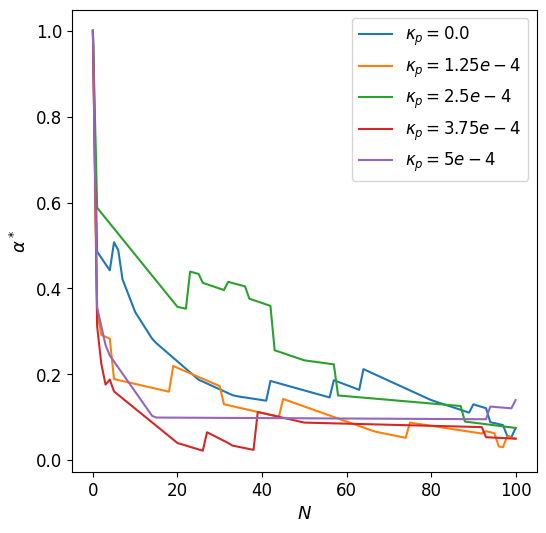}
    \caption{The effect of $\kappa_t =[0.0, 2.5\times10^{-8}, 5\times10^{-8}, 7.5\times10^{-8}, 1\times10^{-7}]$ (top) and $\kappa_p=[0.0, 0.000125, 0.00025, 0.000375, 0.0005]$ (bottom) on $\alpha^*$ for agent $1$ (left) and agent $2$(right). We used the average of $10$ networks for this simulation.}
    \label{fig:sensitivity_history}
\end{figure}

To observe the effects of $\kappa_p$ and $\kappa_t$, we fix $\phi=0.0$ and $\gamma=6$. We choose to fix $\phi=0.0$ to avoid influence of other agents to the dynamics. In Figure \ref{fig:sensitivity_history}, we plot $\alpha^*$ over different time-steps $N$ for different $\kappa_t$ (top) and $\kappa_p$ (bottom) for both agents.

The temporary price impact, $\kappa_\tau$, is the local impact each agent has on the asset price due to the velocity of trade. To maximize the trade criteria, the agents take a less aggressive approach and prefer to liquid their positions in a more gradual manner. 
We observe that when the permanent price impact, $\kappa_p$, is low the agent trades more actively and the activity decreases as $\kappa_p$ increases. This is expected because as the permanent price impact increases, the more impact the agent has on the asset price. Since the goal of the optimal execution problem is to trade the asset without affecting the price as much as possible the agent will be more hesitant to trade as their permanent price impact increases. 

\begin{figure}[htbp]
    \centering
    \includegraphics[width=0.36\textwidth]{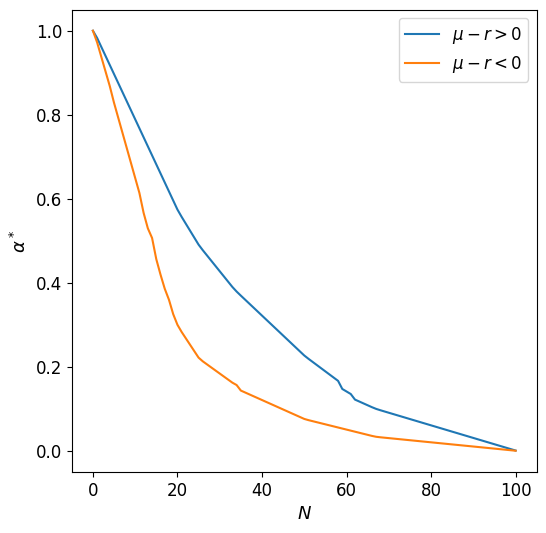}
    \hspace{0.4cm}
    \includegraphics[width=0.36\textwidth]{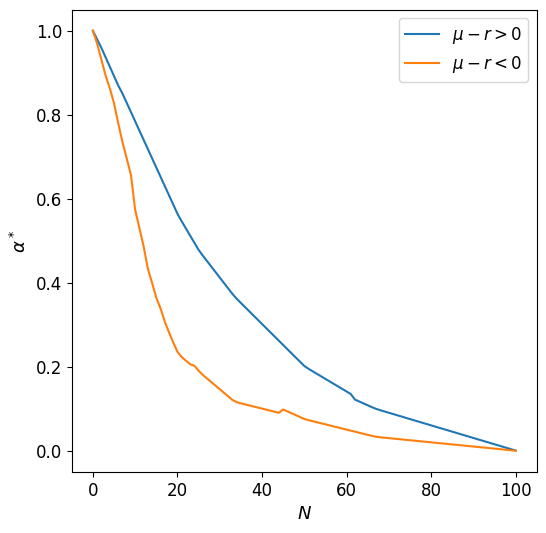}
    \includegraphics[width=0.37\textwidth]{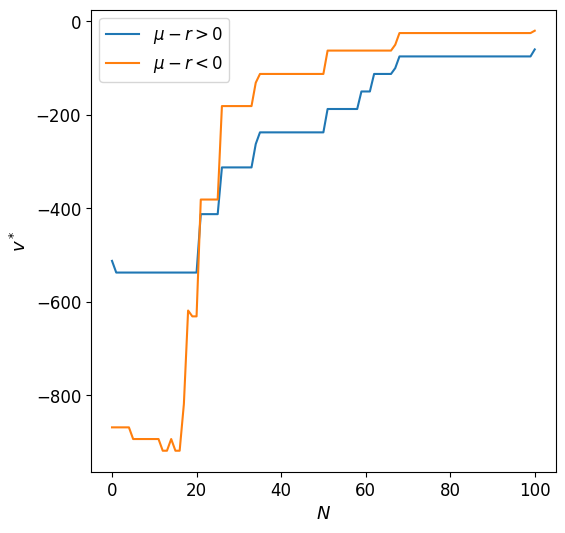}
    \hspace{0.3cm}
    \includegraphics[width=0.37\textwidth]{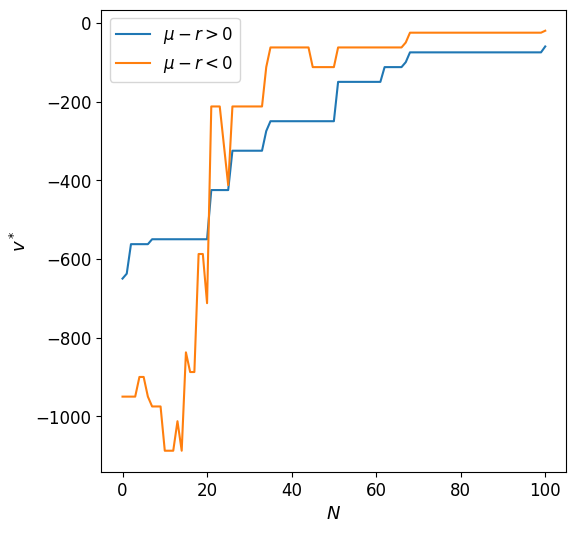}
    \caption{$\alpha^*$ for agent $1$ (top left) and agent $2$ (top right). $v^*$ for agent $1$ (bottom left) and agent $2$ (bottom right) with $\kappa_t=1\times 10^{-7}$. We fix $M=1000$ and $N=100$. We used the average of $100$ networks for this simulation.}
    \label{fig:portfolio_good}
\end{figure}

\subsection{Experiment 3: Optimal Execution in Advantageous and Disadvantageous Market Conditions for $2$ Agents and the Effects of Correlation}

In this experiment, we study the affects of macroeconomic pressure on the asset. We consider $K=2$, two agents that sell off $d=1$ asset and consider two cases. To simulate good market conditions, we set $\mu-r>0$. In contrast, to simulate poor market conditions we set $\mu-r<0$. We assume there is no permanent price impact, i.e. $\kappa_p = 0.0$, $\phi = 0.2$ and $\gamma = 6$. For good market conditions, we set $\mu = 0.1$ and $r = 0.05$. For poor market conditions, we set the $\mu = 0.0$. Figure \ref{fig:portfolio_good} shows the sell only strategy of agent $1$ (left) and agent $2$ (right). The portfolio is presented in the top figures and the controls are presented in the bottom figures.

In Figure \ref{fig:portfolio_good}, we observe that the agents liquidate their position faster if they expect the price of the asset to decrease. This is expected but how they choose to sell is important to observe. We can see the optimal trade speed for $\mu-r>0$ is more gradual. This means the sell off is more gradual and we do not observe large jump in behaviour. This is not the case for $\mu-r<0$ as we see fast sell off in the first $20$ timesteps. Then any residual amount of the asset is sold off more gradually until $T$. Note the difference in optimal paths is due to randomization of the training data during the batch training.

\begin{figure}[htbp]
    \centering
    \includegraphics[width=0.35\columnwidth]{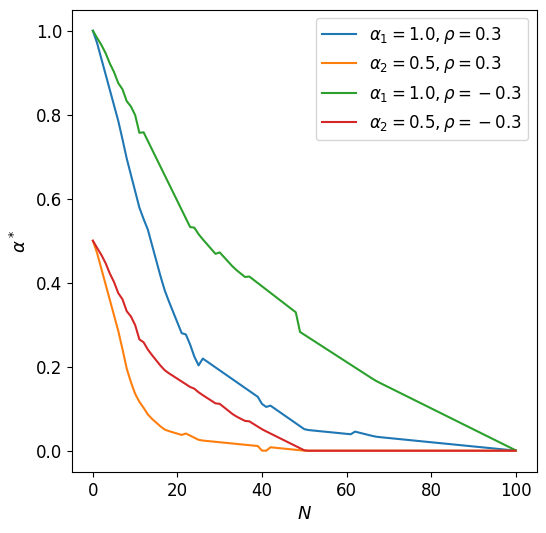}
    \hspace{0.3cm}
    \includegraphics[width=0.35\columnwidth]{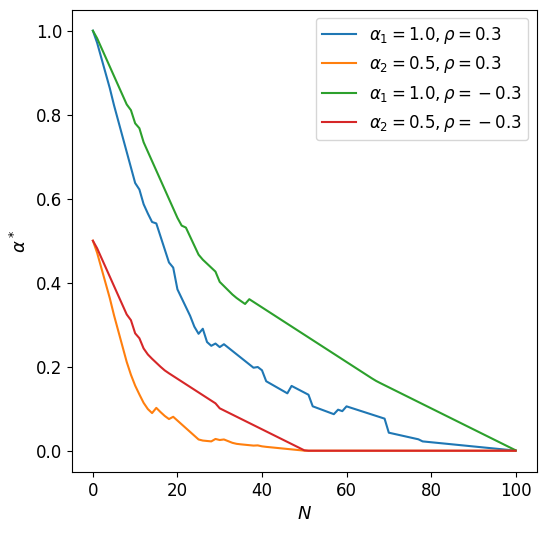}
    \includegraphics[width=0.35\columnwidth]{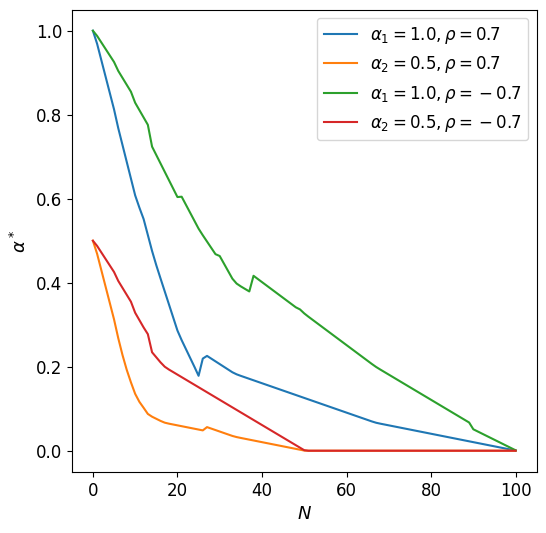}
    \hspace{0.3cm}
    \includegraphics[width=0.35\columnwidth]{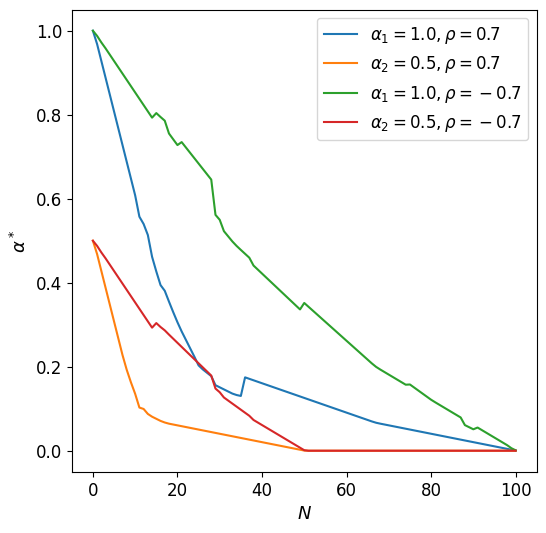}
    \caption{$\alpha^*$ for agent 1 (top left) and agent 2 (top right) with low correlation. $\alpha^*$ for agent 1 (bottom left) and agent 2 (bottom right) with high correlation.We used the average of $100$ networks for this experiment.}
    \label{fig:portfolio_good_2k_2d}
\end{figure}

We extend the experiment to $d=2$ to study the affects of correlation with the proposed method. Since we directly consider the number of shares of an asset held instead of the proportion, we do not need to worry about adding additional constraints to the problem. In experiment $2$, we only look at the selling agents and consider $4$ cases to study the affects of correlation in the order execution problem. In case $1$, we simulate good market conditions $(\mu_1-r>0$, $\mu_2-r>0)$. In case $2$, we simulate poor market conditions $(\mu_1-r<0$, $\mu_2-r<0)$. In case $3$, we simulate mixed market conditions $(\mu_1-r>0$, $\mu_2-r<0)$. We study the affects of correlation between assets $1$ and $2$. We set weak positive correlation as $\rho=0.3$ and strong positive correlation as $\rho=0.7$. We set weak negative correlation as $\rho=-0.3$ and strong negative correlation as $\rho=-0.7$.

\vskip2mm
\noindent\textit{Case $\mathbf{1}$: $\mathbf{\mu_1-r > 0, \mu_1-r>0}$}\vskip2mm

In case $1$, we set $\mu_1=\mu_2=0.1$ and $r=0.05$. Figure \ref{fig:portfolio_good_2k_2d} shows the optimal portfolio for each asset with high and low correlations. We also explore the case of positive and negative correlation. The top row shows the optimal portfolio for agent $1$ (left) and agent $2$ (right) with weak correlation. The bottom left and right figures show the optimal portfolio of the agents for strong correlation. Due to the non-uniqueness of the optimal trade rate, there is no uniqueness in the optimal portfolios.

We observe that when assets are negatively correlated, the agents tend to hold more risky assets than when they are positively correlated. Positively correlated assets sell more quickly because when one asset is sold, the price impact increases the value of correlated assets, prompting the agent to sell off assets faster.  Asset $1$ is held longer to reduce the overall impact when the agent liquidates their position. Figure \ref{fig:control_good_2k_2d} shows the optimal control for each agent in assets with low correlation (top left and right) and high correlation (bottom left and right).

\begin{figure}[htbp]
    \centering
    \includegraphics[width=0.36\columnwidth]{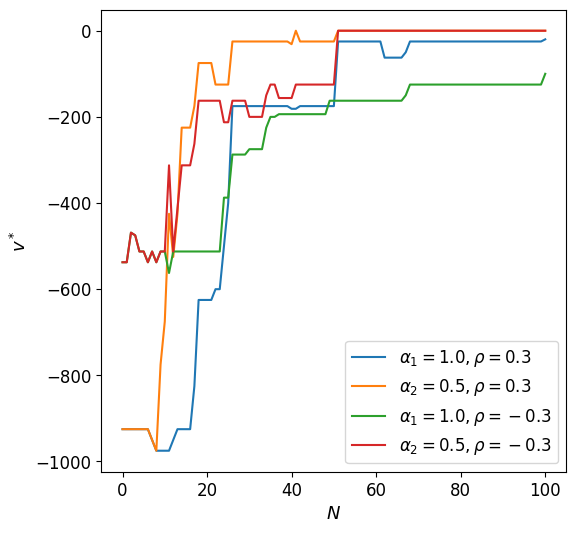}
    \hspace{0.3cm}
    \includegraphics[width=0.36\columnwidth]{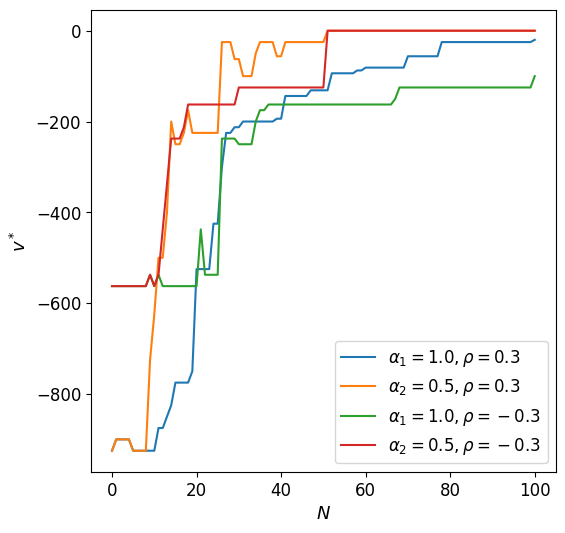}
    \includegraphics[width=0.36\columnwidth]{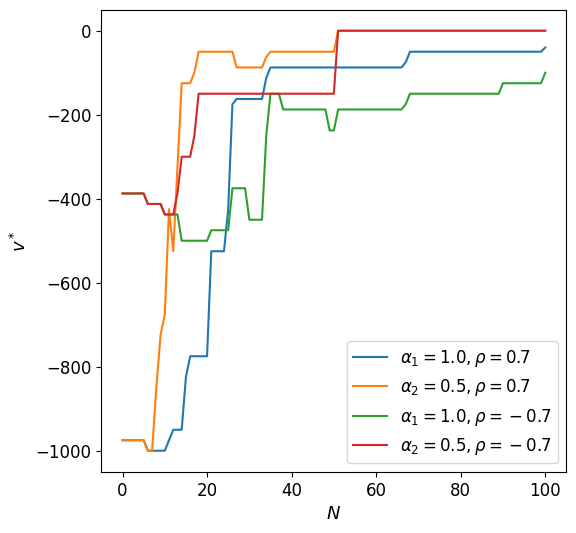}
    \hspace{0.3cm}
    \includegraphics[width=0.36\columnwidth]{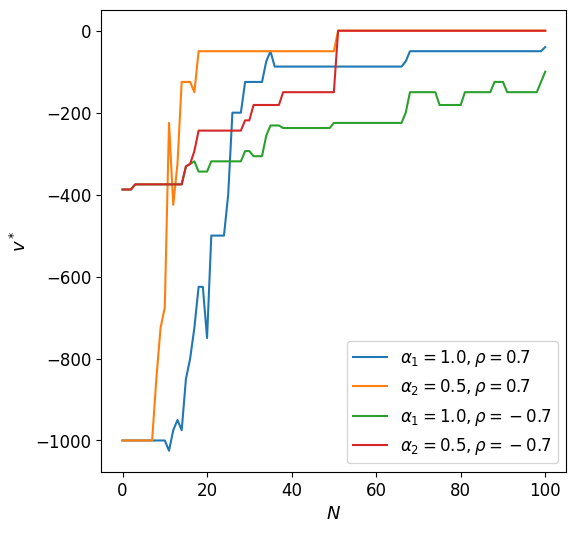}
    \caption{$v^*$ for agent 1 (top left) and agent 2 (top right) with low correlation. $v^*$ for agent 1 (bottom left) and agent 2 (bottom right) with high correlation. We used the average of $100$ networks for this experiment.}
    \label{fig:control_good_2k_2d}
\end{figure}

In Figure \ref{fig:control_good_2k_2d}, we observe that the trade rates of each agent is more negative for positively correlated assets which implies that they are selling at a faster rate. In contrast, for negatively correlated assets, the trade rate is less negative resulting in slower liquidation of the risky assets. This results in a larger position in the risky asset held over time. Figure \ref{fig:control_good_2k_2d} shows the same trade rate for assets $1$ and $2$ at $N=0$. This implies that the agents do not have a strong preference on holding asset $1$ or $2$ when they first begin to trade.

\vskip2mm
\noindent\textit{Case $\mathbf{2}$: $\mathbf{\mu_1-r < 0, \mu_2-r < 0}$}\vskip2mm

In case $2$, we set $\mu_1=\mu_2=0.0$ and $r=0.05$. Case $2$ assumes that both assets are under poor market conditions. Figure \ref{fig:portfolio_bad_2k_2d} plots the optimal portfolio for agents $1$ and $2$ when the assets are weakly and strongly correlated. The corresponding optimal controls are plotted in Figure \ref{fig:control_bad_2k_2d}.

\begin{figure}[htbp]
    \centering
    \includegraphics[width=0.36\columnwidth]{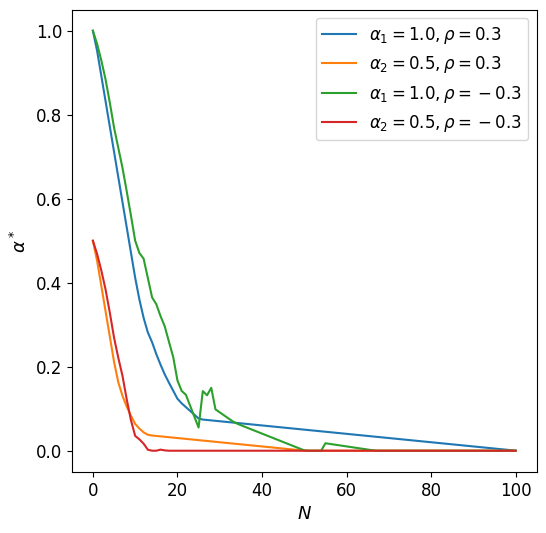}
    \hspace{0.3cm}
    \includegraphics[width=0.36\columnwidth]{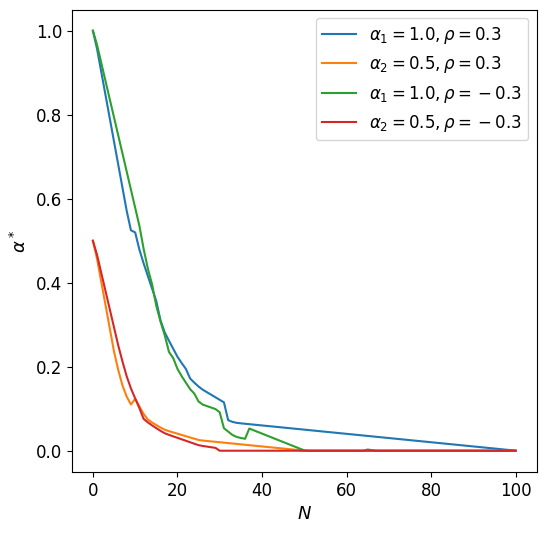}
    \includegraphics[width=0.36\columnwidth]{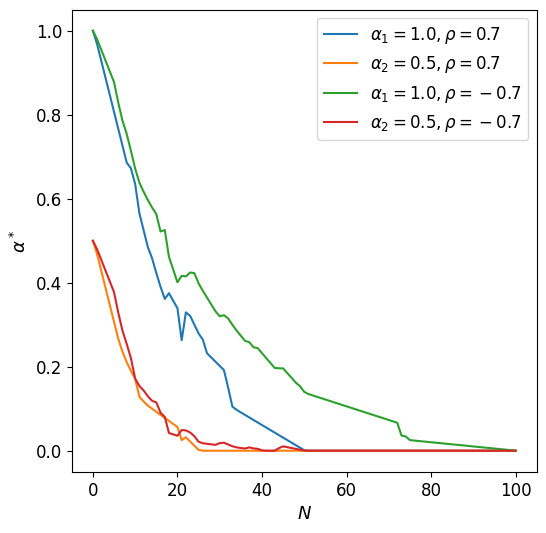}
    \hspace{0.3cm}
    \includegraphics[width=0.36\columnwidth]{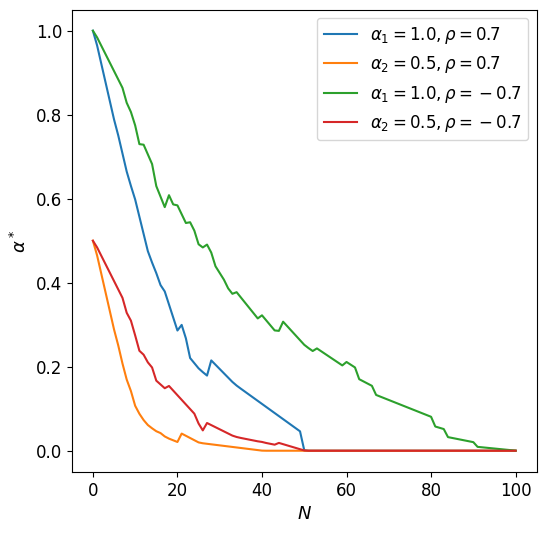}
    \caption{$\alpha^*$ for agent 1 (top left) and agent 2 (top right) with low correlation. $\alpha^*$ for agent 1 (bottom left) and agent 2 (bottom right) with high correlation. We used the average of $100$ networks for this experiment.}
    \label{fig:portfolio_bad_2k_2d}
\end{figure}

In Figure \ref{fig:portfolio_bad_2k_2d}, we observe that for weakly correlated assets, the agents actively liquidated assets regardless of correlation. However, under strong correlation negatively correlated, assets are liquidated more gradually. This is in line with what was observed in case $1$. Since the two assets are negatively correlated, the agents expect asset $1$ to increase when asset $2$ decreases, and vice versa. This factor reduces the negative drift experienced by each asset, resulting in the agents liquidating the assets more gradually.

\begin{figure}[htbp]
    \centering
    \includegraphics[width=0.36\textwidth]{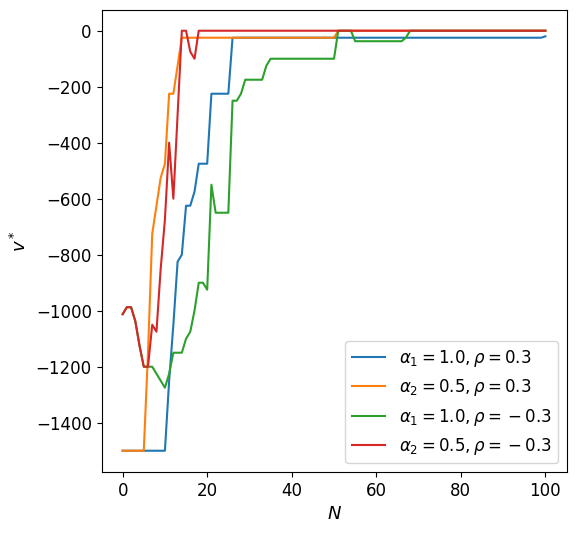}
    \hspace{0.3cm}
    \includegraphics[width=0.36\textwidth]{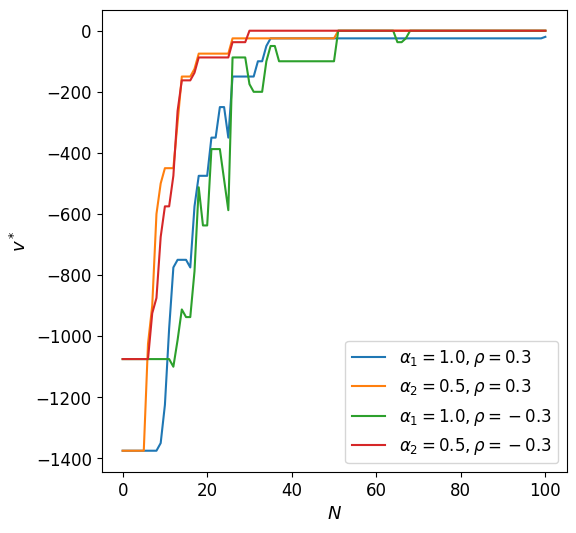}
    \includegraphics[width=0.36\textwidth]{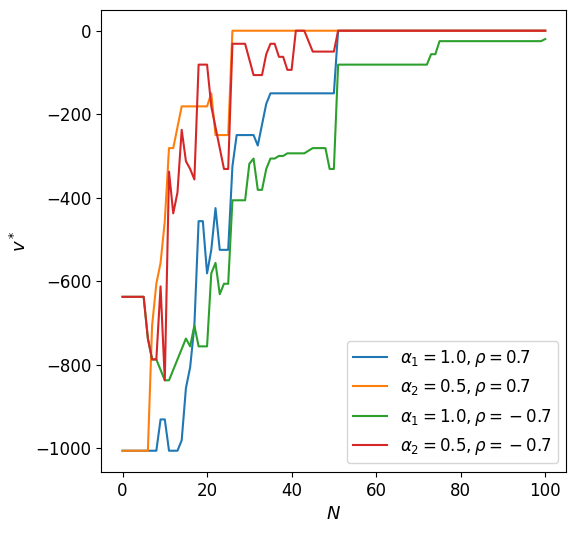}
    \hspace{0.3cm}
    \includegraphics[width=0.36\textwidth]{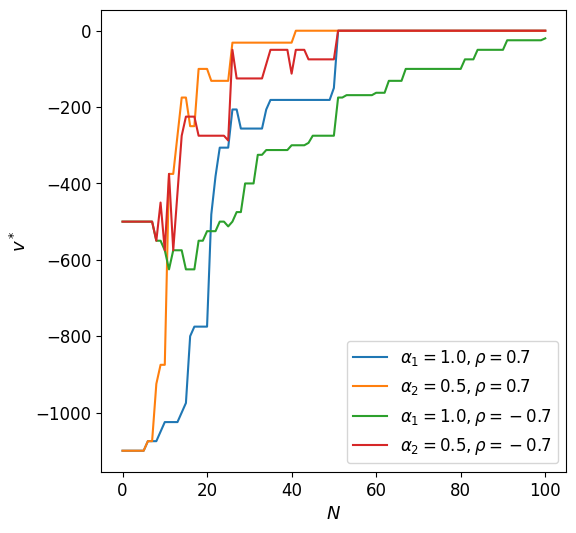}
    \caption{$v^*$ for agent 1 (top left) and agent 2 (top right) with low correlation. $v^*$ for agent 1 (bottom left) and agent 2 (bottom right) with high correlation. We used the average of $100$ networks for this experiment.}
    \label{fig:control_bad_2k_2d}
\end{figure}

In Figure \ref{fig:control_bad_2k_2d}, the agents liquidate positively correlated more quickly at $N=0$. This reflects the initial steeper slope shown for positively correlated assets. Note, the path of the controls is not unique in our problem, as the HJB equation only guarantees uniqueness of the value function.

\vskip2mm
\noindent\textit{Case $\mathbf{3}$: $\mathbf{\mu_1-r > 0, \mu_2-r < 0}$}\vskip2mm

In case $3$, we set $\mu_1=0.1$, $\mu_2=0.0$ and $r=0.05$. Figure \ref{fig:portfolio_good_bad_2k_2d} plots the optimal portfolio of agent $1$ and agent $2$ when $\mu_1-r>0$ and $\mu_2-r<0$. The optimal control of agent $1$ and agent $2$ are shown in Figure \ref{fig:control_good_bad_2k_2d}.

\begin{figure}[htbp]
    \centering
    \includegraphics[width=0.36\textwidth]{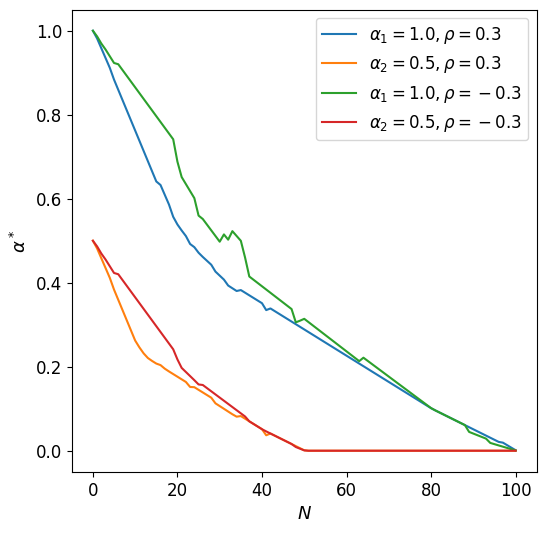}
    \hspace{0.3cm}
    \includegraphics[width=0.36\textwidth]{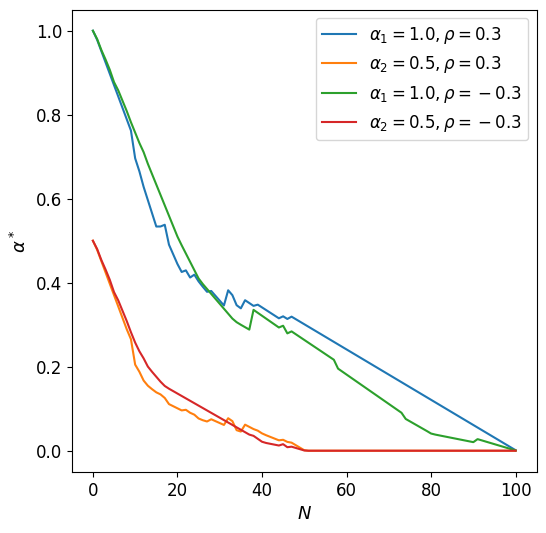}
    \includegraphics[width=0.36\textwidth]{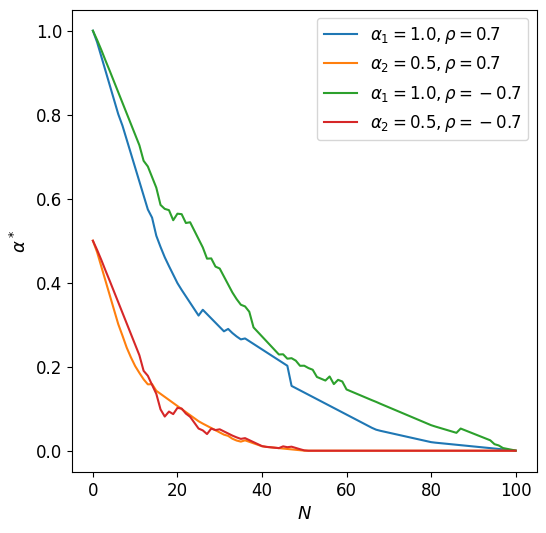}
    \hspace{0.3cm}
    \includegraphics[width=0.36\textwidth]{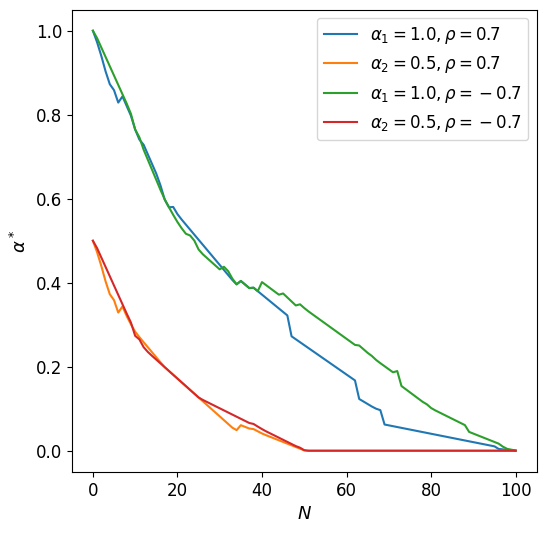}
    \caption{$\alpha^*$ for agent 1 (top left) and agent 2 (top right) with low correlation. $\alpha^*$ for agent 1 (bottom left) and agent 2 (bottom right) with high correlation. We used the average of $100$ networks for this experiment.}
    \label{fig:portfolio_good_bad_2k_2d}
\end{figure}

In Figure \ref{fig:portfolio_good_bad_2k_2d}, we observe that when there are mixed market conditions, the agents try to liquidate the assets in a similar manner regardless of correlation. Both negative and positive correlation encourages the agents to gradually liquidate asset $1$ as asset $1$ is expected to increase its price over time. In Figure \ref{fig:control_good_bad_2k_2d}, we see the agents quickly liquidating asset $2$. We also observe that the agents actually do not liquidate asset $1$ completely until the assets are sold off at $T$. 

\subsection{Experiment 4: Optimal Execution for 2 Agents with Different Price Impact}
In this experiment, we study the interaction of two agents with different permanent price impact factors, $\kappa_{P1}, \kappa_{P2}$. Agent $1$ has $\kappa_{P1} = 5\times 10^{-4}$ and agent $2$ has $\kappa_{P2} = 0.0$. This signifies that agent $1$ is a very influential trading agent who has a lot of impact on price dynamics. Conversely, agent $2$ is a trading agent with no price impact at all. We study the interaction between the agents for $1$ asset and observe how they interact in a good economy,  i.e. $\mu-r > 0$ and in a bad economy, i.e. $\mu-r < 0$. We also consider the different levels of performance awareness, $\phi = 0.0, 0.3, 0.7$. We set $r = 0.05$, $\mu = 0.1$ or $\mu = 0.0$, $s_0 = 1.0$, $\alpha = 0.0$, $\kappa_\tau = 0.0$, and $\kappa_s = 0.0$ for $T = 1/250$, $M =1000$, $N=100$. 

\begin{figure}[htbp]
    \centering
    \includegraphics[width=0.36\textwidth]{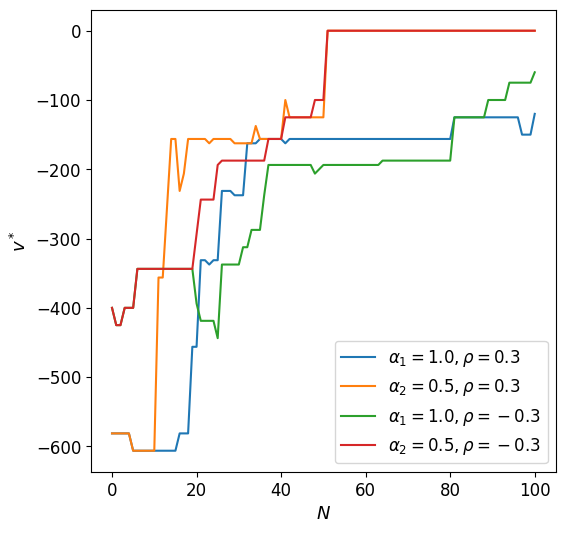}
    \hspace{0.3cm}
    \includegraphics[width=0.36\textwidth]{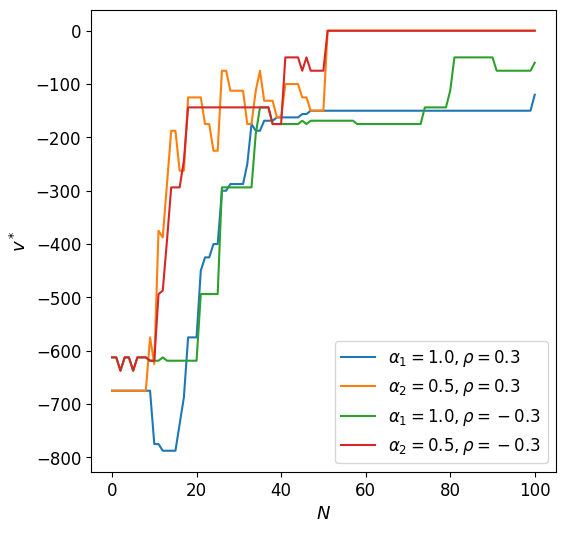}
    \includegraphics[width=0.36\textwidth]{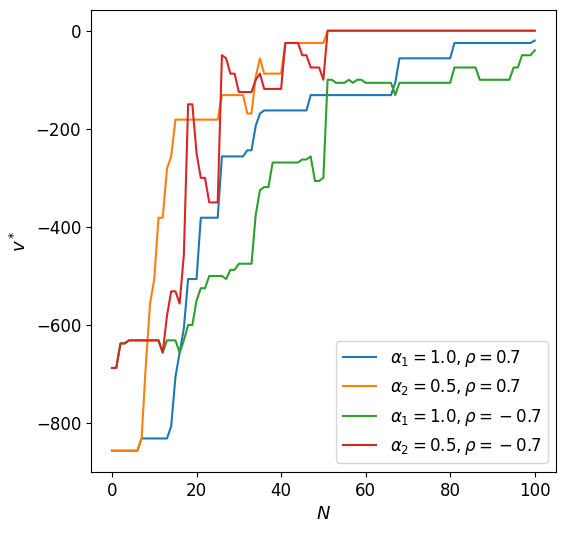}
    \hspace{0.3cm}
    \includegraphics[width=0.36\textwidth]{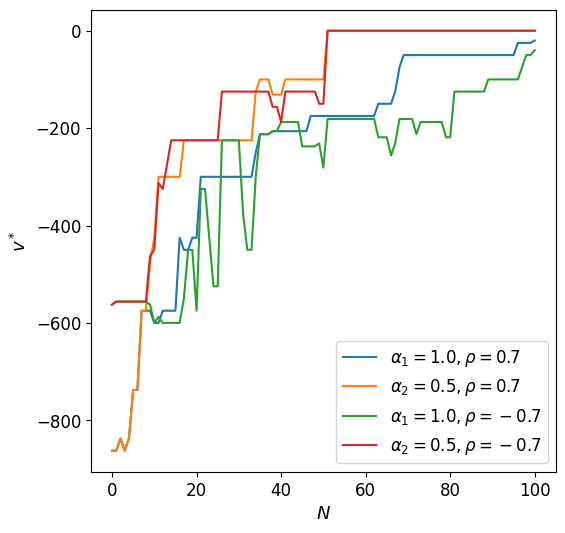}
    \caption{$v^*$ for agent 1 (top left) and agent 2 (top right) with low correlation. $v^*$ for agent 1 (bottom left) and agent 2 (bottom right) with high correlation. We used the average of $100$ networks for this experiment.}
    \label{fig:control_good_bad_2k_2d}
\end{figure}

First, we consider the case when $\phi=0.0$. The optimal portfolio and trade rate for agent $1$ and agent $2$ are shown in Figure \ref{fig:whale_phi_0}. We observe that when $\mu-r > 0$, agent $1$ is very inactive in their trading activity and holds a large proportion of wealth in risky assets. Agent $2$ on the other hand sells off their position consistently until around $N=20$.

The trade rates tell us that agent $2$ perceives the market to be more risky than agent $1$. This can be seen in the trade rates, where agent $2$ sells off more quickly than agent $1$. When $\mu - r < 0$, we observe an interesting behaviour. Agent $2$ immediately sells off their position. There is no penalty for agent $2$ to do this as they have no impact on the asset price. However, agent $1$ does not do the same and actually holds a large position in the risky asset longer before liquidating fully. 

\begin{figure}[htbp]
    \centering
    \includegraphics[width = 0.36\textwidth]{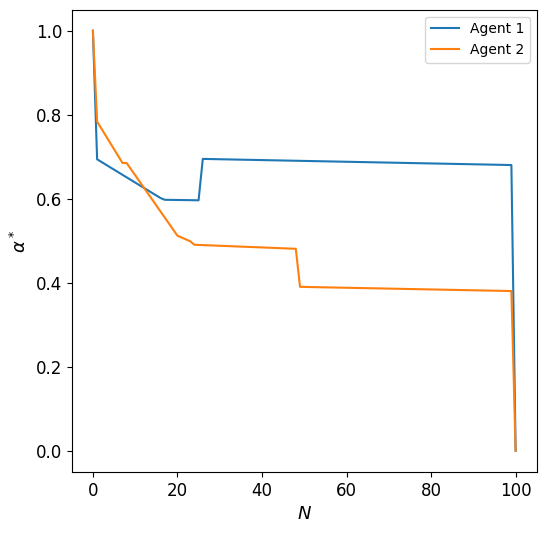}
    \hspace{0.3cm}
    \includegraphics[width = 0.36\textwidth]{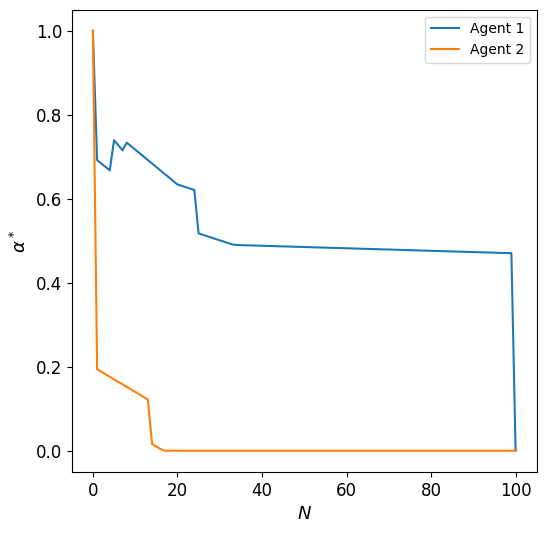}
    \includegraphics[width = 0.36\textwidth]{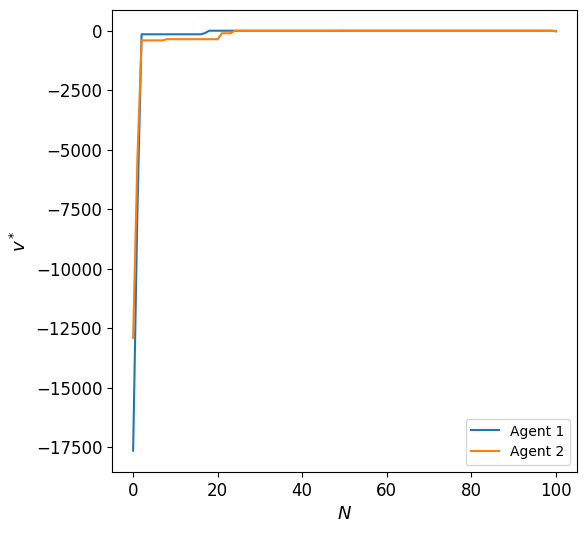}
    \hspace{0.3cm}
    \includegraphics[width = 0.36\textwidth]{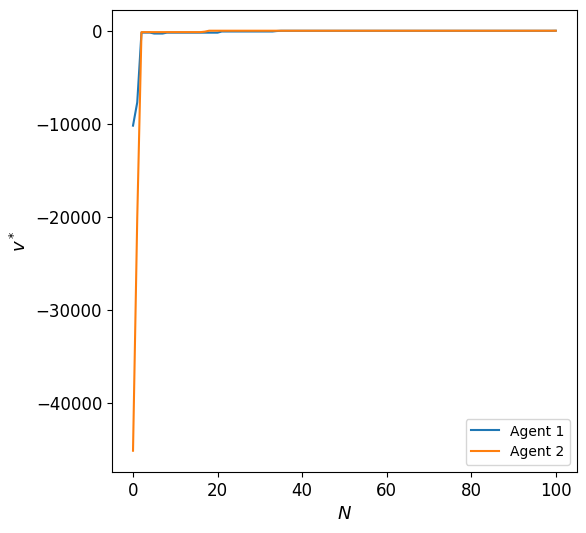}
    \caption{The optimal portfolio (top) and trade rate (bottom) for Agents $1$ (blue) and $2$ (orange). For a good economy (right) and bad economy (left) when $\phi=0.0$.}
    \label{fig:whale_phi_0}
\end{figure}

Next, we consider the case when $\phi=0.3$. The optimal portfolio and trade rate for agent $1$ and agent $2$ are shown Figure \ref{fig:whale_phi_0_3}. In this case, both agents are aware of the general performance of the other agent and they both sell quickly. When $\mu-r > 0$, both agents want to hold more positions in the risky asset until the moment they want to liquidate after $N=100$, this is expected as the payoff from holding the asset until $N=100$ outweighs the penalty from failing to liquidate in the time horizon when $\mu - r > 0$. We see that agent $1$ trades faster than agent $2$ until around $N = 20$, but the trading eases as $N$ goes to $100$. When $\mu - r < 0$, we see a similar behaviour to the good economy case for agent $1$ except that agent $1$ trades faster until around $N = 10$. We see agent $2$ trade more gradually until $N=10$ as it would see the asset price drop due to the actions of agent $1$.

Finally, we consider the case when $\phi=0.7$. In Figure \ref{fig:whale_phi_0_7}, we observe that agent $1$ is willing to hold more position in the risky asset but not as quickly as when $\phi = 0.0$. This is because both agents now weigh the average performance of the market almost as much as their own individual performance. When $\mu-r > 0$, again both agents want to hold more positions in the risky asset until the moment they want to liquidate after $N=100$. We see that agent $2$ trades faster than agent $1$ at all timesteps. When $\mu - r < 0$, agent $1$ sells more actively than agent $2$ which in turn lowers the asset price. In response, agent $2$ reduces their position in the risky asset drastically from the start.

\begin{figure}[htbp]
    \centering
    \includegraphics[width = 0.36\textwidth]{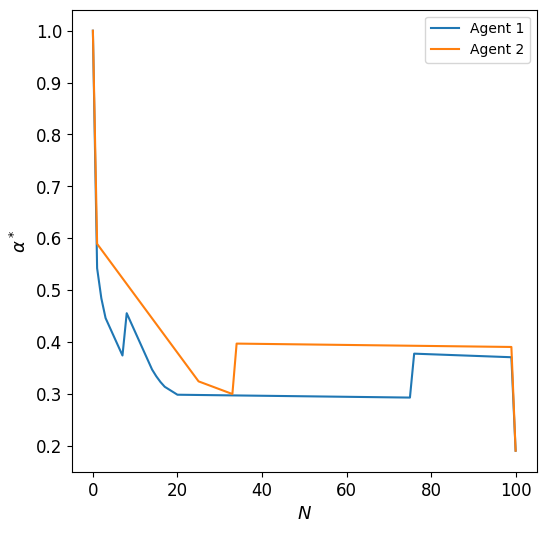}
    \hspace{0.3cm}
    \includegraphics[width = 0.36\textwidth]{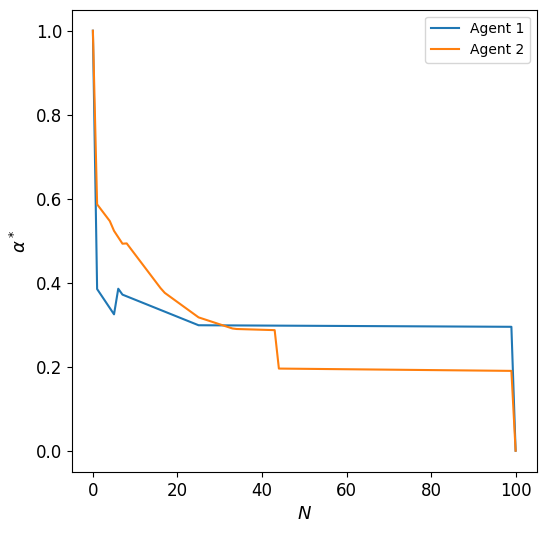}
    \includegraphics[width = 0.36\textwidth]{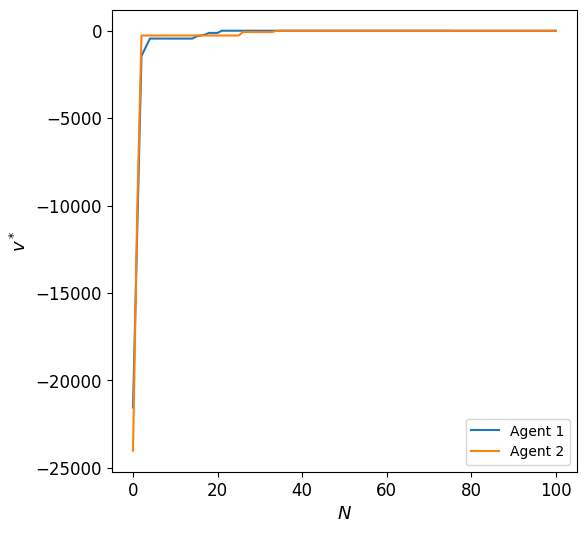}
    \hspace{0.3cm}
    \includegraphics[width = 0.36\textwidth]{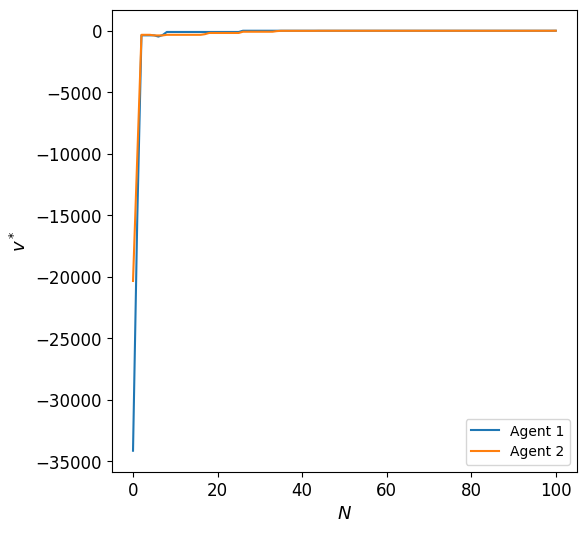}
    \caption{The optimal portfolio (top) and trade rate (bottom) for Agents $1$ (blue) and $2$ (orange). For a good economy (right) and bad economy (left) when $\phi=0.3$.}
    \label{fig:whale_phi_0_3}
\end{figure}

\subsection{Experiment 5: Optimal Execution for Multiple Buyers and Sellers}

In this experiment, we compare the investment performance between sellers (long) and buyers (short) in different economic conditions. To compare different groups of agents, we compute the average Sharpe ratio of the selling agents and the buying agents. The Sharpe ratio provides an insight into the relative performance of the portfolio return. A positive Sharpe ratio is considered an investment strategy that outperforms the risk free rate. A negative Sharpe ratio indicates the liquidating strategy; i.e. selling and buying, have lower return than holding the asset. In our problem, a strategy with Sharpe ratio $> 0$ is acceptable. 

For each agent $k=1,...,K$, the return on liquidating the position is given by $R_k=(X_k(T)-X_k(0))/X_k(0)$ and the standard deviation on the return is $\sigma_{R_k}=\sqrt{Var(R_k)}$. We want to compare the strategy with the strategy of holding the asset from $t=0$ to $T$. The no-trade return is given by $R_{0,k}=(\sum_i\alpha_{0,i}S_i(T)-X_k(0))/X_k(0)$. Thus the Sharpe ratio is given by:
\[
    SR_k = \frac{\mathbb{E}[R_k]-\mathbb{E}[R_{0,k}]}{\sigma_{R_k}}.
\]
Note we compare the portfolio with the expected value of the asset instead of the risk free rate of return. This compares the performance of different strategies versus holding the asset without liquidating. If the number of long agents is given by $K_{long}$ and the number of short agents is given by $K_{short}$, then we can compute
\[
    \sum_{k\in K_{long}} SR_k/K_{long}
    \;\;\;\; \mbox{and} \;\;\; \sum_{k\in K_{short}} SR_k/K_{short}
\]
to measure the average performance of the long agents (buyers) and short agents (sellers) respectively.

\begin{figure}[htbp]
    \centering
    \includegraphics[width = 0.36\textwidth]{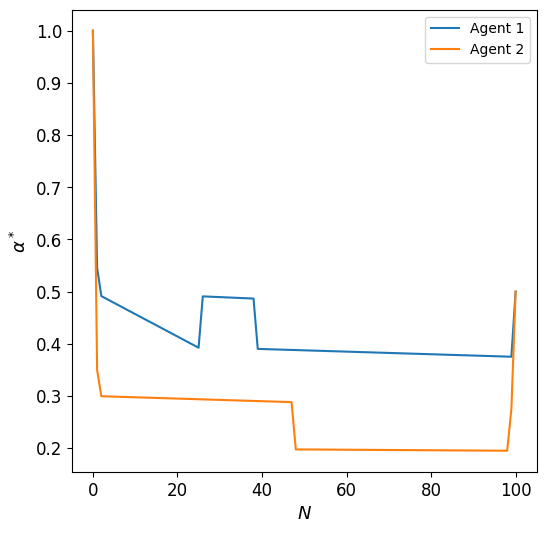}
    \hspace{0.3cm}
    \includegraphics[width = 0.36\textwidth]{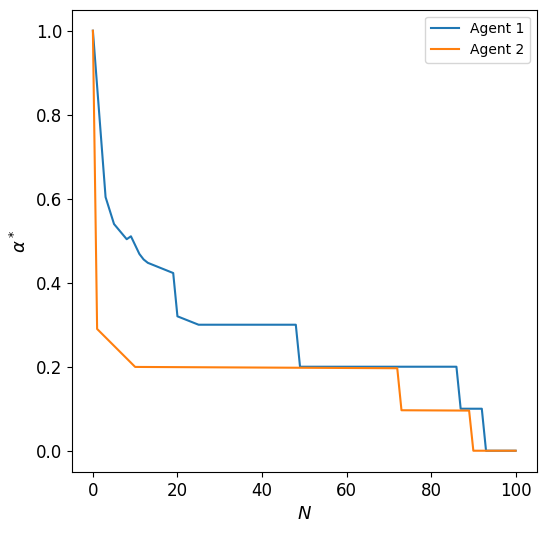}
    \includegraphics[width = 0.37\textwidth]{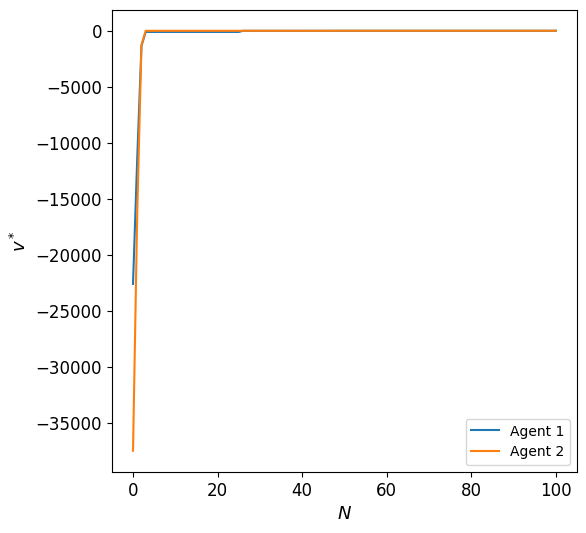}
    \hspace{0.3cm}
    \includegraphics[width = 0.37\textwidth]{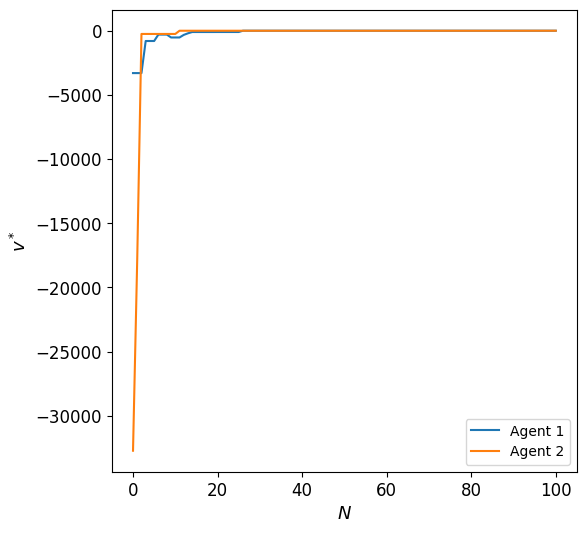}
    \caption{The optimal portfolio (top) and trade rate (bottom) for Agents $1$ (blue) and $2$ (orange). For a good economy (right) and bad economy (left) when $\phi=0.7$.}
    \label{fig:whale_phi_0_7}
\end{figure}

We simulate $K = 10$ agents and $d = 2$ assets with $\rho=0$. We do not consider correlation in this example so that the results are easier to understand. In this experiment, we set $r = 0.05$ for $i=1,...,10$, $\kappa_\tau = 1e-7$, $\kappa_p=0.000125$, $\gamma=2$ and fix $\phi = 0.2$. By increasing the number of agents beyond $K>2$, we can study more interesting dynamics between agents. In case $1$, we let $7$ agents be selling agents and $3$ agents be buying agents. Selling agents are agents liquidating long positions and buying agents are agents liquidating short positions. In case $2$, we look at the $5$ selling and $5$ buying agents. In case $3$, we look at the case where we have $3$ selling agents and $7$ buying agents. For all three cases, we consider different market conditions:
\begin{enumerate}
    \item good market conditions $(\mu_1-r>0,\mu_2-r>0)$,
    \item poor market conditions $(\mu_1-r<0,\mu_2-r<0)$,
    \item mixed market conditions $(\mu_1-r>0,\mu_2-r<0)$.
\end{enumerate}

\vskip2mm
\noindent\textit{Case $1$: $7$ selling agents, $3$ buying agents}

In case $1$, we look at the average Sharpe ratio of $7$ selling agents and $3$ buying agents in different market conditions. The average Sharpe ratio for sellers and buyers are summarized in Table \ref{tab:7_sell_3_buy}. We observe that under market conditions $2)$ and $3)$, it is advantageous for selling and buying agents to liquidate. The best Sharpe ratio we achieve is $0.35014$ for selling agents under $2)$. This is because the agents liquidate the short positions gradually and the long positions quickly. There is an overall negative price impact due to the asymmetry in the number of sellers and buyers which leads to lower asset prices. Under market conditions $1)$, it is disadvantageous to liquidate short positions over holding the asset. This occurs from the overall negative price impact that decreases the amount of growth for both assets. Due to the imbalance of sellers and buyers, we observe the selling agent always performing better than the holding strategy.\\

\begin{table}[htbp]
    \centering
    \begin{tabular}{c|c|c}\hline 
        Market Condition & $\sum_{k\in K_{long}}SR_k/K_{long}$& $\sum_{k\in K_{short}}SR_k/K_{short}$ \\ \hline 
        $\mu_1-r>0,\mu_2-r>0$ & $0.01456$ & $-0.02949$\\ \hline
        $\mu_1-r<0,\mu_2-r<0$ & $0.35014$ & $0.14856$ \\ \hline 
        $\mu_1-r>0,\mu_2-r<0$ & $0.06993$ & $ 0.04475$\\ \hline
    \end{tabular}
    \caption{Average Sharpe ratio of sellers and buyers in good market conditions, poor market conditions and mixed market conditions.}
    \label{tab:7_sell_3_buy}
\end{table}

\vskip2mm
\noindent\textit{Case $2$: $5$ selling agents, $5$ buying agents}\vskip2mm

In case $2$, we look at the average Sharpe ratio of $5$ selling agents and $5$ buying agents in different market conditions. The average Sharpe ratio for sellers and buyers are summarized in Table \ref{tab:5_sell_5_buy}.

\begin{table}[htbp]
    \centering
    \begin{tabular}{c|c|c}\hline 
        Market Condition & $\sum_{k\in K_{long}}SR_k/K_{long}$& $\sum_{k\in K_{short}}SR_k/K_{short}$ \\ \hline 
        $\mu_1-r>0,\mu_2-r>0$ & $0.00909$ & $-0.03696$\\ \hline
        $\mu_1-r<0,\mu_2-r<0$ & $0.15290$ & $0.06123$ \\ \hline 
        $\mu_1-r>0,\mu_2-r<0$ & $0.1814$ & $0.1080$\\ \hline
    \end{tabular}
    \caption{Average Sharpe ratio of sellers and buyers in good market conditions, poor market conditions and mixed market conditions.}
    \label{tab:5_sell_5_buy}
\end{table}

We observe that under market conditions $2)$ and $3)$, it is advantageous for long and short agents to liquidate. The best Sharpe ratio we achieve is $0.1814$ for sellers under $3)$. This is because agents sell (buy) asset $1$ more gradually (quickly) and asset $2$ more quickly (gradually). Because $\alpha_{1,k}>\alpha_{2,k}$, the agents gain more from the increase in asset $1$ price than the decrease in asset $2$ price. This results in an advantage for gradual sellers and aggressive buyers. Under $1)$, sellers have an advantage since they sell off gradually and the price of both assets are expected to increase. Buyers are disadvantaged as the expected value of the asset increases under $1)$, even if they liquidate quickly.

\vskip2mm
\noindent\textit{Case $3$: $3$ selling agents, $7$ buying agents}\vskip2mm

In case $3$, we look at the average Sharpe ratio of $3$ selling agents and $7$ buying agents in market conditions $1), 2)$, and $3)$. The average Sharpe ratio for sellers and buyers are summarized in Table \ref{tab:3_sell_7_buy}. We observe that under market conditions $2)$, it is advantageous for sellers and buyers to liquidate. The best Sharpe ratio we achieve is $0.3113$ for selling agents under $2)$. There is an overall positive price impact since the number of buyers is greater than the number of sellers. We see that sellers and buyers generally do not out perform the hold strategy. This is due to the overall positive price impact adding to the drift term of both assets pushing up the asset price.

\begin{table}[htbp]
    \centering
    \begin{tabular}{c|c|c}\hline 
        Market Condition & $\sum_{k\in K_{long}}SR_k/K_{long}$& $\sum_{k\in K_{short}}SR_k/K_{short}$ \\ \hline 
        $\mu_1-r>0,\mu_2-r>0$ & $-0.06766$ & $-0.14143$\\ \hline
        $\mu_1-r<0,\mu_2-r<0$ & $0.3113$ & $0.1212$ \\ \hline 
        $\mu_1-r>0,\mu_2-r<0$ & $-0.02911$ & $-0.06120$\\ \hline
    \end{tabular}
    \caption{Average Sharpe ratio of sellers and buyers in good market conditions, poor market conditions and mixed market conditions.}
    \label{tab:3_sell_7_buy}
\end{table}

\section{Conclusion}

In this paper, we extended the optimal trade execution problem under the Markowitz criteria to a time-consistent optimal trade execution problem for multiple assets and multiple agents setting. This is done by formulating the original mean-variance problem into its dual representation coupled with an auxiliary HJB equation \citep{aivaliotis-2018}. To extend this into higher dimensions, we reformulate the auxiliary HJB equation into a BSDE that can be solved using a neural network. We present a deep residual self-attention network to solve the BSDE presented in \eqref{eq:BSDE_discrete_U}. This allows us to use a deep residual network that learns temporal dependencies. This overcomes the issue of saving weights for every timestep $N$ and overcomes the vanishing gradient problem of recurrent networks \citep{na-2023}. 

We show in Experiment $1$ that the proposed method matches classical solutions where applicable. In Experiment $2$ we show sensitivities of the proposed method with respect to $\phi, \gamma, \kappa_t, \kappa_p$. When an agent is more conscious other agents in the market, i.e. $\phi$ increases, agents tend to invest more in risky assets. When an agent is more risk adverse, i.e. $\gamma$ increases, the agent invests less in risky assets. We saw that as the temporary price impact of the agent increases they tend to liquidate more gradually and as the permanent price impact increases the agents liquidate faster. 

In Experiment $3$, we looked at the optimal execution of $2$ agents under different market conditions. With $d=1$ we saw that the agents liquidated assets gradually when the market conditions were good and liquidated quickly when they were poor. We also looked at the affects of strong and weak correlation with $d=2$. In Experiment $4$, we study the interaction between $2$ agents, where one agent contributed all the price impact. We found that in general as $\phi$ increase we see decreased trading activity from agent $2$, we also see agent $1$ hold less risky assets. In good economic conditions agent $1$ holds larger positions in the risky asset and is less active in trading than in poor economic conditions. 

In Experiment $5$, we studied the average performance of a group of sellers and buyers in different economic conditions in a market consisting of $K=10$ agents. We found that in most economic conditions active trading on both sides of the market benefited sellers and buyers. In any scenario, when the economy was good, buyers underperformed the hold position. In the case of $3$ sellers and $7$ buyers both sellers and buyers underperformed the holding strategy, except when both assets are expected to reduce in price. 

In this paper, we attempted to cover as much as we can. However, there are many different aspects left to explore. The model we used can be updated to include models for limit order books and trade volumes \citep{alfonsi-2010,cartea-2016}.

\section*{Disclosure Statement}
No potential conflict of interest is reported by the author(s).

\section*{Funding}
This work was supported by the Natural Sciences and Engineering Research Council of Canada.

\section*{ORCID}
\begin{description}
\item \textit{Andrew Na}: https://orcid.org/
0000-0002-6162-8171
\item \textit{Justin W.L. Wan}: https://orcid.org/
0000-0001-8367-6337
\end{description}

\bibliographystyle{plainnat}
\bibliography{References}

\end{document}